\documentclass{lmcs}
\pdfoutput=1

\usepackage{lastpage}
\renewcommand{\dOi}{14(4:7)2018}
\lmcsheading{}{1--\pageref{LastPage}}{}{}%
{Oct.~31,~2017}{Oct.~29,~2018}{}

\usepackage{xcolor}
\definecolor{linkblue}{RGB}{0,0,180}

\usepackage{tikz}
\usetikzlibrary{plotmarks}
\usepackage{xspace}
\usepackage{amssymb}
\usepackage[section]{placeins}

\usepackage[noend]{algorithmic}
\usepackage{newfloat}
\DeclareFloatingEnvironment[name=Listing]{listing}
\makeatletter
\let\c@listing\c@figure
\makeatother

\usepackage{proof}%lkproof on CTAN
\newcommand\iname[1]{\textit{#1}}
\newcommand\irule[1]{\iname{#1}}
\newcommand\inferr[1][]{\infer[\irule{\strut #1}]}

\theoremstyle{plain}
\newtheorem{theorem}[thm]{Theorem}
\newtheorem{lemma}[thm]{Lemma}
\newtheorem{proposition}[thm]{Proposition}

\theoremstyle{definition}
\newtheorem{definition}[thm]{Definition}

\newtheorem{remark}[thm]{Remark}
\newtheorem{example}[thm]{Example}

\usepackage{mathtools} % extensible arrows

\newcommand\MnFont[1]{
 \DeclareFontFamily{U}{MnSymbol#1}{}
 \DeclareSymbolFont{MnSy#1}{U}{MnSymbol#1}{m}{n}
 \SetSymbolFont{MnSy#1}{bold}{U}{MnSymbol#1}{b}{n}
 \DeclareFontShape{U}{MnSymbol#1}{m}{n}{
     <-6>  MnSymbol#15
    <6-7>  MnSymbol#16
    <7-8>  MnSymbol#17
    <8-9>  MnSymbol#18
    <9-10> MnSymbol#19
   <10-12> MnSymbol#110
   <12->   MnSymbol#112}{}
 \DeclareFontShape{U}{MnSymbol#1}{b}{n}{
     <-6>  MnSymbol#1-Bold5
    <6-7>  MnSymbol#1-Bold6
    <7-8>  MnSymbol#1-Bold7
    <8-9>  MnSymbol#1-Bold8
    <9-10> MnSymbol#1-Bold9
   <10-12> MnSymbol#1-Bold10
   <12->   MnSymbol#1-Bold12}{}
}
\MnFont{A}\MnFont{B}\MnFont{C}\MnFont{D}
\DeclareMathSymbol{\mnvdash}{\mathrel}{MnSyA}{'330}
\DeclareMathSymbol{\mndashv}{\mathrel}{MnSyA}{'332}
\DeclareMathSymbol{\mncirc}{\mathbin}{MnSyC}{'130}
\DeclareMathSymbol{\mnshortparallel}{\mathbin}{MnSyC}{'376}
\DeclareMathSymbol{\medslash}{\mathbin}{MnSyC}{'022}

\let\smallparallel\mnshortparallel
\let\smalldevel\mncirc

\newcommand{\overlayrel}[4]{\mathrel{%
 \def\next##1##2{%
  \setbox0=\hbox{$##1#3$}%
  \setbox1=\hbox to\wd0{$##1\hfil\mkern#1mu#4\mkern#2mu\hfil$}%
  \dp1=\dp0\ht1=\ht0\wd0=0pt\box0\box1%
 }%
 \mathpalette\next{}%
}}
\newcommand{\fakerel}[3]{\mathrel{%
 \def\next##1##2{%
  \setbox0=\hbox{$##1#1$}%
  \setbox1=\hbox to\wd0{$##1#2\hss#3$}%
  \dp1=\dp0\ht1=\ht0\box1%
 }%
 \mathpalette\next{}%
}}

\newcommand{\prightarrow}{\overlayrel02\rightarrow\smallparallel}
\newcommand{\pleftarrow}{\overlayrel20\leftarrow\smallparallel}
\newcommand{\xprightarrow}[2][]%
 {\overlayrel02{\xrightarrow[#1]{#2}}\smallparallel}
\newcommand{\xpleftarrow}[2][]%
 {\overlayrel20{\xleftarrow[#1]{#2}}\smallparallel}
\newcommand{\orightarrow}{\overlayrel02\rightarrow\smalldevel}
\newcommand{\oleftarrow}{\overlayrel20\leftarrow\smalldevel}
\newcommand{\xorightarrow}[2][]%
 {\overlayrel02{\xrightarrow[#1]{#2}}\smalldevel}
\newcommand{\xoleftarrow}[2][]%
 {\overlayrel20{\xleftarrow[#1]{#2}}\smalldevel}
\newcommand{\uleftrightarrow}{\fakerel{\to}{\mnvdash}{\mndashv}}
\newcommand{\uleftarrow}{\fakerel{\to}{\mnvdash}{\relbar}}
\newcommand{\urightarrow}{\fakerel{\to}{\relbar}{\mndashv}}
\makeatletter
\newcommand{\xuleftrightarrow}[2][]{\ext@arrow3095\@uarrow@fill{#1}{#2}}
\def\@uarrow@fill{\arrowfill@\uleftarrow\relbar\urightarrow}
\renewcommand{\xleftrightharpoons}[2][]{\mathrel{%
 \ext@arrow 0095\@leftrightharpoons@fill{#1}{#2}}}
\def\@leftrightharpoons@fill{\arrowfill@\leftharpoonup\relbar\rightharpoondown}
\def\@leftrightharpoons@fill{\arrowfill@%
 {\smash{\raise.22ex\hbox to 0pt{$\leftharpoonup$}}\@rreq}%
 \Relb@r%
 {\smash{\lower.22ex\hbox to 0pt{$\rightharpoondown$}}\@rreq}%
}
\def\Relb@r{\smash{\raise.22ex\hbox to 0pt{$\relbar$}\lower.22ex\hbox{$\relbar$}}}
\def\@rreq{\fakerel\to\Relb@r\Relb@r}

\newcommand{\r@rrow}[3]{%
 \newcommand{#1}[2][]{%
  \def\next{#2\@ifempty{##1}{}{_{##1}}\@ifempty{##2}{}{^{##2}}}%
  \mathchoice{#3[##1]{##2}}{\next}{\next}{\next}%
 }%
}
\newcommand{\l@rrow}[3]{%
 \newcommand{#1}[2][]{%
  \def\next####1{%
   \setbox0=\hbox{$####1\vphantom{#2}\@ifempty{##1}{}{_{\vphantom{##1}}}%
    \@ifempty{##2}{}{^{##2}}$}%
   \setbox1=\hbox{$####1\vphantom{#2}\@ifempty{##1}{}{_{##1}}%
    \@ifempty{##2}{}{^{\vphantom{##2}}}$}%
   \setbox2=\vbox{\hbox to\wd0{}\hbox to\wd1{}}%
   \mathrel{\hskip\wd2\hskip-\wd0\box0\hskip-\wd1\box1{#2}}%
  }%
  \mathchoice{#3[##1]{##2}}{\next\textstyle}%
   {\next\scriptstyle}{\next\scriptscriptstyle}%
 }%
}
\r@rrow{\xr}{\rightarrow}{\xrightarrow}
\r@rrow{\xlr}{\leftrightarrow}{\xleftrightarrow}
\l@rrow{\xl}{\leftarrow}{\xleftarrow}
\r@rrow{\xR}{\Rightarrow}{\xRightarrow}
\r@rrow{\xLR}{\Leftrightarrow}{\xLeftrightarrow}
\l@rrow{\xL}{\Leftarrow}{\xLeftarrow}
\r@rrow{\Xr}{\prightarrow}{\xprightarrow}
\l@rrow{\Xl}{\pleftarrow}{\xpleftarrow}
\r@rrow{\XR}{\orightarrow}{\xorightarrow}
\l@rrow{\XL}{\oleftarrow}{\xoleftarrow}
\r@rrow{\xrw}{\mathrel{\widetilde\Rightarrow}}{\xrightarrow}
\r@rrow{\xrb}{\Rightarrow}{\xrightarrow}
\l@rrow{\xlw}{\mathrel{\widetilde\leftarrow}}{\xleftarrow}
\r@rrow{\xu}{\uleftrightarrow}{\xuleftrightarrow}
\l@rrow{\xlu}{\leftharpoonup}{\xleftharpoonup}
\r@rrow{\xru}{\rightharpoonup}{\xrightharpoonup}
\l@rrow{\xld}{\leftharpoondown}{\xleftharpoondown}
\r@rrow{\xrd}{\rightharpoondown}{\xrightharpoondown}

\newcommand{\xln}[1][]{\xl[#1]{{\mkern-3mu!}\hfill}}
\newcommand{\xrn}[1][]{\xr[#1]{\hfill!\mkern-3mu}}
\makeatother

\newcommand{\x}[1]{\mathcal{#1}}
\newcommand{\m}[1]{\mathsf{#1}}
\newcommand{\CR}{\text{CR}\xspace}
\newcommand{\NFP}{\text{NFP}\xspace}
\newcommand{\UNC}{\text{UNC}\xspace}
\newcommand{\UNR}{\text{UNR}\xspace}
\newcommand\TS{\mathsf{TS}}
\newcommand{\NF}{\mathsf{NF}}
\newcommand{\NNF}{\neg\mathsf{NF}^\circ}
\newcommand\nf[1][]{{\downarrow}_{#1}}

\newcommand{\OO}{\mathcal{O}}
\newcommand{\BB}{\x{B}}
\newcommand{\NN}{\x{N}}
\newcommand{\CC}{\x{C}}
\newcommand{\EE}{\x{E}}
\newcommand{\NNR}{\NN}%{\NN_\RR}
\newcommand{\NNU}{\NN_\UU}
\newcommand{\NNV}{\NN_\VV}
\newcommand{\CCR}{\CC}%{\CC_\RR}
\newcommand{\CCU}{\CC_\UU}
\newcommand{\CCV}{\CC_\VV}
\newcommand{\EER}{\EE}%{\EE_\RR}
\newcommand{\EEU}{\EE_\UU}
\newcommand{\EEV}{\EE_\VV}
\newcommand{\RCR}{\RC}%{\RC_\RR}
\newcommand{\RCU}{\RC_\UU}
\newcommand{\RCV}{\RC_\VV}
\newcommand{\EQU}{_\RR}
\newcommand{\FFR}{\FF\EQU}%{\FF_\RR}
\newcommand{\FF}{\Sigma} %note that \x{F} would clash with \RC
\newcommand{\RR}{\x{R}}
\renewcommand{\SS}{\x{S}}
\newcommand{\UU}{\x{U}}
\newcommand{\VV}{\x{V}}
\newcommand{\RC}{\x{F}}
\newcommand{\TT}{\x{T}}

\newcommand{\CSI}{\textsf{CSI}\xspace}
\newcommand{\FORT}{\textsf{FORT}\xspace}

\newcommand{\nimplies}{\overlayrel02\implies{/}}

\newcommand{\mglue}{\mkern-1.5mu}
\newcommand{\mfa}{\m{f{\mglue}a}}
\newcommand{\mfb}{\m{f{\mglue}b}}

\newcommand{\mffb}{\m{f{\mglue}f{\mglue}b}}

\newcommand{\mfffb}{\m{f{\mglue}f{\mglue}f{\mglue}b}}

\begin{document}

\newcommand{\TRSs}{TRSs}
\title[Deciding \CR, \NFP, \UNC, and \UNR for Ground \TRSs]{Deciding Confluence and Normal Form Properties of Ground Term Rewrite Systems Efficiently}

\author[Bertram Felgenhauer]{Bertram Felgenhauer}	%required
\address{Department of Computer Science, University of Innsbruck, Technikerstrasse 21a, 6020 Innsbruck}	%required
\email{bertram.felgenhauer@uibk.ac.at}  %optional
\thanks{This research was supported by FWF (Austrian Science Fund) project P27528.}	%optional

%\author[]{Author 2}	%optional
%\address{address2; addresses should be duplicated when authors share an affiliation}	%optional
%\email{author2@email2; ditto for email addresses}  %optional
%\thanks{thanks 2, optional.}	%optional

%% etc.

%% required for running head on odd and even pages, use suitable
%% abbreviations in case of long titles and many authors:

%% mandatory lists of keywords and classifications:
\titlecomment{The decision procedure for confluence of ground term rewrite systems previously appeared at RTA 2012.}
%%%%%%%%%%%%%%%%%%%%%%%%%%%%%%%%%%%%%%%%%%%%%%%%%%%%%%%%%%%%%%%%%%%%%%%%%%%%%%

%% the abstract has to PRECEED the command \maketitle:
%% be sure not to issue the \maketitle command twice!

\begin{abstract}
It is known that the first-order theory of rewriting
is decidable for ground term rewrite systems,
but the general technique uses tree automata
and often takes exponential time.
For many properties,
including confluence (\CR),
uniqueness of normal forms with respect to reductions (\UNR)
and with respect to conversions (\UNC),
polynomial time decision procedures are known for ground term rewrite systems.
However, this is not the case for the normal form property (\NFP).
In this work, we present a cubic time algorithm for \NFP,
an almost cubic time algorithm for \UNR,
and an almost linear time algorithm for \UNC,
improving previous bounds.
We also present a cubic time algorithm for \CR.
\end{abstract}

\maketitle

%% start the paper here:
\section{Introduction}
\label{sec:intro}

In this article, we consider four properties of finite ground term rewrite systems,
that is, first-order term rewrite systems (TRSs) without variables.
These properties are

\begin{itemize}
\item
confluence or, equivalently, the Church-Rosser property (\CR),
which states that any two convertible terms have a common reduct;
\item
the normal form property (\NFP),
which holds if any term convertible to a normal form can be reduced
to that normal form;
\item
uniqueness of normal forms with respect to conversions (\UNC),
meaning that any two convertible normal forms are equal; and
\item
uniqueness of normal forms with respect to reductions (\UNR),
stating that from any term, at most one normal form can be reached.
\end{itemize}

In seminal work~\cite{DT85},
Dauchet and Tison established that
the first-order theory of ground term rewrite systems is decidable
using tree automata techniques.
This result is applicable to all four properties.
While the procedure is usually exponential,
it yields polynomial time procedures for \UNC and \UNR with a bit of care.
This is elaborated in Section~\ref{sec:rel:dt}.

In fact it is known that \CR, \UNC, and \UNR are decidable in polynomial time for ground TRSs.
In this article,
we are interested in bounding the exponent of the associated polynomials,
which is of importance to implementers.

As far as we know, the best previous result for \UNC is an almost quadratic algorithm
by Verma et al.~\cite{VRL01} with $\OO(\Vert\RR\Vert^2\log\Vert\RR\Vert)$ time complexity,
where $\Vert\RR\Vert$ denotes the sum of the sizes of the rules of $\RR$,
and the size of a rule is the sum of the sizes
of its left-hand side and right-hand side.
In Section~\ref{sec:unc} we present an algorithm that decides \UNC in $\OO(\Vert\RR\Vert\log\Vert\RR\Vert)$ time.
In the case of \UNR for ground TRSs,
Verma~\cite{V09} and Godoy and Jacquemard~\cite{GJ09} have established
that polynomial time algorithms exist, using tree automata techniques.
No precise bound is given by these authors.
In Section~\ref{sec:unr} we present an $\OO(\Vert\RR\Vert^3\log\Vert\RR\Vert)$ time algorithm for deciding \UNR.
Furthermore we present a $\OO(\Vert\RR\Vert^3)$ decision procedure for \NFP
for ground TRSs, which will be covered in Section~\ref{sec:nfp}.
As far as we know, this is the first polynomial time decision procedure
for \NFP in the literature.
Last but not least,
we present a $\OO(\Vert\RR\Vert^3)$ decision procedure for \CR,
closely based on previous work by the author~\cite{F12}.
See Section~\ref{sec:cr} for details.
In Section~\ref{sec:rel} we discuss related work,
and we conclude in Section~\ref{sec:end}.

%%%%%%%%%%%%%%%%%%%%%%%%%%%%%%%%%%%%%%%%%%%%%%%%%%%%%%%%%%%%%%%%%%%%%%%%%%%%%%
\section{Preliminaries}
\label{sec:prel}

We assume familiarity with term rewriting and (bottom-up) tree automata.
For an overview of term rewriting, see~\cite{BN98};
for tree automata, please refer to~\cite{tata}.
We recall the notions used in this article.

A signature is a set of function symbols
$\FF$ each associated with an arity (which are natural numbers).
The ground terms $\TT(\FF)$ over $\FF$ are constructed inductively in
the usual way: If $t_1,\dots,t_n \in \TT(\FF)$ and
$f \in \FF$ has arity $n$,
then $f(t_1,\dots,t_n) \in \TT(\FF)$. A position $p$ of a term is a
sequence of natural numbers addressing a subterm $t|_p$.
Replacement of subterms $t[u]_p$, and the size of terms $|t|$ have
their standard definitions~\cite{BN98}. A term $t$ together with a
position $p$ defines a context $C[\cdot] = t[\cdot]_p$.
Contexts can be instantiated, $C[s] = t[s]_p$.
Alternatively, contexts $C$ can be viewed as terms that contain an
extra constant $\square$, representing a hole, exactly once.
Then $C[t]$ denotes the result of replacing $\square$ by $t$.
Multi-hole contexts are terms that may contain several holes;
for a multi-hole context $C$ with $n$ occurrences of $\square$,
$C[t_1,\dots,t_n]$ denotes the term obtained by replacing the holes
in $C$ by $t_1$ to $t_n$ from left to right.
Function symbols with arity $0$ are called constants.
A ground term is flat if it is either a constant
or a function symbol applied to constants.

A set $\RR \subseteq \TT(\FF)^2$ of rules is a ground term rewrite system (ground TRS).
If $(\ell,r) \in \RR$, we also write $\ell \to r \in \RR$.
By $\RR^-$, $\RR^\pm$, $|\RR|$, $\Vert\RR\Vert$ we denote $\RR^{-1}$
(i.e., the inverse of $\RR$, where we view $\RR$ as a relation on ground terms),
$\RR \cup \RR^{-1}$, the number of rules in $\RR$, and
the total size of the rules, $\sum_{\ell \to r \in \RR} (|\ell| + |r|)$,
respectively. Any ground TRS $\RR$ induces a rewrite relation $\xr[\RR]{}$
on ground terms:
$s \xr[\RR]{} t$ if there is a context $C[\cdot]$ such that
$C[\ell] = s$ and $C[r] = t$ for some $\ell \to r \in \RR$.
Properties like flatness extend to rules and TRSs.
For example, a rule is left-flat if its left-hand side is flat; a TRS is
left-flat if all its rules are.
We write $t \trianglelefteq \RR$ if $t$ is a subterm of a side of a rule in $\RR$.

Given a rewrite relation $\xr{}$, we denote by $\xl{}$, $\xlr{}$,
$\xr{=}$, $\xr{*}$, its inverse, symmetric closure,
reflexive closure, reflexive transitive closure, respectively,
and $\cdot$ composes rewrite relations.
A term $s$ is a normal form (w.r.t.\ $\to$)
if there is no $t$ with $s \to t$.
By $\xr{!}$ we denote reduction to normal form: $s \xr{!} t$ if
$s \xr{*} t$ and $t$ is a normal form with respect to $\xr{}$.
Two terms $s$ and $t$ are convertible if $s \xlr{*} t$. They are joinable,
denoted by $s \downarrow t$, if $s \xr{*} \cdot \xl{*} t$.
If we have two rewrite relations, $\xr[1]{}$ and $\xr[2]{}$, then
${\xr[1]{}}/{\xr[2]{}}$ is defined as $\xr[2]{*}{\cdot}\xr[1]{}{\cdot}\xr[2]{*}$;
applied to TRSs, we write $\RR/\SS$ for the relative TRS
that induces ${\xr[\RR]{}}/{\xr[\SS]{}}$ as a rewrite relation.
Given a signature $\FF$, the parallel closure $\xr{\mnshortparallel}$
is defined inductively by
\[
\infer{s \xr{\mnshortparallel} s}{} \qquad
\infer{s \xr{\mnshortparallel} t}{s \xr{} t} \qquad
\infer{f(s_1,\dots,s_n) \xr{\mnshortparallel} f(t_1,\dots,t_n)}{s_1 \xr{\mnshortparallel} t_1 & \dots &
s_n \xr{\mnshortparallel} t_n}
\]
Let us recall the four main properties of interest for this paper.
A rewrite relation $\xr{}$ has \dots
\begin{itemize}
\item
the Church-Rosser property (\CR),
if $s \xlr{*} t$ implies $s \xr{*} {\cdot} \xl{*} t$ for all $s$ and $t$;
\item
the normal form property (\NFP),
if $s \xlr{*} t$ with $t$ in normal form implies $s \xr{*} t$ for all $s$ and $t$;
\item
unique normal forms with respect to conversion (\UNC),
if $s \xlr{*} t$ implies $s = t$
for all normal forms $s$ and $t$; and
\item
unique normal forms with respect to reductions (\UNR),
if  $t \xl{!} s \xr{!} u$ implies $t = u$ for all $s$, $t$ and $u$.
\end{itemize}
We refer to the last three properties collectively as \emph{normal form properties}.
It is well known (and easy to see) that
\[\CR \implies \NFP \implies \UNC \implies \UNR\]
The converse implications are false, as demonstrated by the standard examples in Figure~\ref{fig:not-unr-unc-nfp-cr}.
A rewrite relation $\to$ is terminating if there are no infinite rewrite sequences
$t_0 \to t_1 \to \dots$.
For terminating rewrite relations,
confluence and the three normal form properties are equivalent.
If $\to$ is terminating and confluent,
then $s \nf$ denotes the normal form of $s$.

\begin{figure}
\centerline{
\begin{tikzpicture}[scale=0.7]
\node (1a) at (0,0) {$\circ$};
\node (1b) at (1,1.5) {$\cdot$};
\node (1c) at (2,0) {$\circ$};
\node at (1,-1.5) {$\neg\UNR \phantom{\neg}$};
\draw[->] (1b) -- (1a);
\draw[->] (1b) -- (1c);
\node (2a) at (4,0) {$\circ$};
\node (2b) at (5,1.5) {$\cdot$};
\node (2c) at (6,0) {$\cdot$};
\node (2d) at (7,1.5) {$\cdot$};
\node (2e) at (8,0) {$\circ$};
\node at (6,-1.5) {$\phantom{\neg} \UNR \wedge \neg\UNC$};
\draw[->] (2b) -- (2a);
\draw[->] (2b) -- (2c);
\draw[->] (2c) to [distance=10mm,in=-60,out=-120,loop] (2c);
\draw[->] (2d) -- (2c);
\draw[->] (2d) -- (2e);
\node (3a) at (10,0) {$\circ$};
\node (3b) at (11,1.5) {$\cdot$};
\node (3c) at (12,0) {$\cdot$};
\node at (11,-1.5) {$\phantom{\neg} \UNC \wedge \neg\NFP$};
\draw[->] (3b) -- (3a);
\draw[->] (3b) -- (3c);
\draw[->] (3c) to [distance=10mm,in=-60,out=-120,loop] (3c);
\node (4a) at (14,0) {$\cdot$};
\node (4b) at (15,1.5) {$\cdot$};
\node (4c) at (16,0) {$\cdot$};
\node at (15,-1.5) {$\phantom{\neg} \NFP \wedge \neg\CR$};
\draw[->] (4b) -- (4a);
\draw[->] (4b) -- (4c);
\draw[->] (4a) to [distance=10mm,in=-60,out=-120,loop] (4a);
\draw[->] (4c) to [distance=10mm,in=-60,out=-120,loop] (4c);
\node (5a) at (18,0) {$\circ$};
\node (5b) at (18,1.5) {$\cdot$};
\node at (18,-1.5) {$\CR$};
\draw[->] (5b) -- (5a);
\end{tikzpicture}
}
\caption{Examples showing that
$\top \protect\nimplies
\UNR \protect\nimplies
\UNC \protect\nimplies
\NFP \protect\nimplies
\CR \protect\nimplies
\bot$.}
\label{fig:not-unr-unc-nfp-cr}
\end{figure}

A tree automaton $\x A = (Q,Q_f,\Delta)$ consists of a finite set of states $Q$ disjoint from $\FF$,
a set of final states $Q_f \subseteq Q$,
and a set $\Delta$ of transitions $f(q_1,\dots,q_n) \to q$ and $\epsilon$-transitions $p \to q$,
where $f$ is an $n$-ary function symbol and $q_1,\dots,q_n,p,q \in Q$.
A deterministic tree automaton is an automaton without $\epsilon$-transitions
whose transitions have distinct left-hand sides
(we do not require deterministic tree automata to be completely defined).
Note that $\Delta$ can be viewed as a ground TRS over an extended signature that contains $Q$ as constants.
We write $\xr[\x A]{}$ for $\xr[\Delta]{}$,
where we regard the transitions as rewrite rules.
The language accepted by $\x A$ is $L(\x A) = \{ s \mid s \in \TT(\FF) \text{, } q \in Q_f \text{ and }s \xr[\x A]{*} q \}$.

In the complexity analysis we make use of the fact that systems of Horn
clauses can be solved in linear time, see Dowling and Gallier~\cite{DG84}.
Their procedure finds the smallest solution of a set of Horn clauses,
in the sense that as few atoms as possible become true,
in time linear in the total size of the clauses.
This often allows an elegant description of algorithms that compute finite,
inductively defined sets.
For example, the transitive closure of $R \subseteq I \times I$
can be specified by the inference rules
\[
\inferr[base]{(p,q) \in R^+}{(p,q) \in R} \qquad
\inferr[trans]{(p,r) \in R^+}{(p,q) \in R & r \in I & (q,r) \in R^+}
\]
The relation $R$ and the set $I$ are known in advance,
so we treat $(p,q) \in R$  and $r \in I$ as side conditions
that are either true or false.
On the other hand, the relation $R^+$ is unknown,
so we treat $(p,q) \in R^+$ for $p,q \in I$ as atoms whose truth value should be derived by Horn inference.
There are $\OO(|I|^2)$ Horn clauses for (\iname{base}) and
$\OO(|I|^3)$ Horn clauses for (\iname{trans}).
The size of the individual clauses is $\OO(1)$,
so the transitive closure can be computed in cubic time.

Some of the algorithms presented in this article
use maximally shared terms for efficiency;
this idea is also known as hash consing.
In a maximally shared representation,
each ground term $f(t_1,\dots,t_n)$ is represented by a unique identifier
(e.g., a natural number, or a pointer into memory),
which can be mapped to $f$ and the identifiers of $t_1, \dots, t_n$.
In order to maintain maximal sharing,
a lookup table mapping $f$ and identifiers of $t_1, \dots, t_n$
to the identifier of $f(t_1,\dots,t_n)$ is required.
If the arity of $f$ is bounded,
constructing maximally shared terms incurs a logarithmic overhead
compared to a direct construction.
Crucially though, comparing two maximally shared terms takes constant time.

%%%%%%%%%%%%%%%%%%%%%%%%%%%%%%%%%%%%%%%%%%%%%%%%%%%%%%%%%%%%%%%%%%%%%%%%%%%%%%
\section{Common Elements}
\label{sec:common}

\begin{figure}
\centering
\begin{tikzpicture}[
  -latex,
  sibling distance=4cm,
  every node/.style={shape=rectangle,draw}
]
\node[rounded corners] (trs) at (0,0) {\strut input TRS ($\RR$)};
\node (cu) at (0,-1.5) {\strut \hyperref[sec:curry]{currying} ($\RR^\circ$)};
\node (fl) at (0,-3) {\strut \hyperref[sec:flat]{flattening} ($\RR^\flat, \EER$)};
\node (nf) at (4.5,-4.5) {\strut \hyperref[sec:nf]{normal forms} ($\NNR$)};
\node (cc) at (-4.5,-4.5) {\strut \hyperref[sec:cc]{congruence closure} ($\CCR$)};
\node (rc) at (0,-4.5) {\strut \hyperref[sec:rc]{rewrite closure} ($\RCR, \EER$)};
\node[rounded corners] (cr) at (-4.5,-6) {\strut \hyperref[sec:cr]{\CR}};
\node[rounded corners] (nfp) at (-1.5,-6) {\strut \hyperref[sec:nfp]{\NFP}};
\node[rounded corners] (unc) at (1.5,-6) {\strut \hyperref[sec:unc]{\UNC}};
\node[rounded corners] (unr) at (4.5,-6) {\strut \hyperref[sec:unr]{\UNR}};
\draw (trs) to (cu);
\draw (cu) to (fl);
\draw (cu) to [bend left=15] (nf);
\draw (fl) to (nf);
\draw (cu) to [bend right=15] (cc);
\draw (fl) to (cc);
\draw (fl) to (rc);
\draw (rc) to (cr);
\draw (rc) to (unr);
\draw (rc) to (nfp);
\draw (nf) to (nfp);
\draw (nf) to (unc);
\draw (nf) to (unr);
\draw (cc) to (nfp);
\draw (cc) to (unc);
\draw (cc) to (cr);
\draw[dashed] (unc) to [bend left=10] (nfp);
\end{tikzpicture}
\caption{Dependencies of preprocessing steps.}
\label{fig:common}
\end{figure}

In this section, we present computations that are shared between
the decision procedures for the four properties \UNR, \UNC, \NFP and \CR.
The dependencies are as follows (see also Figure~\ref{fig:common}).
\begin{itemize}
\item
Currying (Section~\ref{sec:curry}) and flattening (Section~\ref{sec:flat})
are preparatory steps
used for all four properties.
\item
In Section~\ref{sec:nf},
we construct an automaton recognizing normal forms,
which is used for the three normal form properties \UNR, \UNC, and \NFP.
\item
Congruence closure (Section~\ref{sec:cc}) is used for \UNC, \NFP, and \CR.
\item
Rewrite closure (Section~\ref{sec:rc})
features in the procedures for \UNR, \NFP and \CR.
\end{itemize}
The dashed arrow from \UNC to \NFP in Figure~\ref{fig:common} indicates that
the procedure for \NFP is an extension of the procedure for \UNC.

It is known that \CR and \UNC are preserved by signature extension
for any TRS
(this is a consequence of \CR and \UNC being modular~\cite{M90,T87}).
The same holds for \NFP and \UNR for left-linear TRSs~\cite{M96,M90}.
Since we are concerned with ground systems,
which are trivially left-linear,
this means that we may assume that the signature of the input TRS
consists exactly of the symbols occurring in the input TRS.
In particular, that signature is finite,
and its size is bounded by $\Vert \RR \Vert$.

\begin{remark}
Neither \NFP nor \UNR are preserved by signature extension in general.
Counterexamples have been given by Kennaway et al.~\cite{KKSV96}.
(It is noteworthy that these counterexamples are presented as
counterexamples to the preservation of \NFP and \UNR by currying;
the failure of signature extension is only mentioned in passing.)
\end{remark}

Relatedly,
the properties \CR, \NFP, \UNC, \UNR are normally defined
on terms with variables instead of ground terms;
the variants where rewriting is restricted to ground terms
are called ground-\CR etc.
Fortunately, for ground TRSs, the addition of variables
makes no difference.
This is because any counterexample to one of these properties
(which is a conversion $s \xlr{*} t$ for \CR, \NFP, or \UNC,
or a peak $t \xl{!} s \xr{!} u$ in the case of \UNR)
includes a counterexample with a root step
(which may be obtained by minimizing the size of $s$),
and any conversion (or peak) with a root step consists solely of
ground terms over the TRS's inherent signature.
(This also shows that for ground TRSs, all four properties
(\CR, \NFP, \UNC, \UNR) are preserved by signature extension.)

\begin{example}[running example]
\label{running0}
We demonstrate the constructions
on the following two ground TRSs $\UU$ and $\VV$.
\begin{align*}
\UU &= \{ \m f(\m a) \to \m a, \m f(\m a) \to \m b, \m a \to \m a \}
\\
\VV &= \{ \m{a} \to \m{b}, \m{a} \to \m{f}(\m{a}),
  \m{b} \to \m{f}(\m{f}(\m{b})), \m{f}(\m{f}(\m{f}(\m{b}))) \to \m{b} \}
\end{align*}
\end{example}

%%%%%%%%%%%%%%%%%%%%%%%%%%%%%%%%%%%%%%%%%%%%%%%%%%%%%%%%%%%%%%%%%%%%%%%%%%%%%%
\subsection{Currying}
\label{sec:curry}

Currying turns an arbitrary TRS into one over constants and a single binary function symbol,
thereby bounding the maximum arity of the resulting TRS.

\begin{definition}
\label{def:curry}
In order to curry a ground TRS $\RR$, we change all function symbols in $\FF$ to be constants,
and add a fresh, binary function symbol $\circ$.
The resulting signature is $\FF^\circ = \FF \cup \{ {\circ} \}$.
We write $\circ$ as a left-associative infix operator
(i.e., $s \circ t$ stands for ${\circ}(s, t)$,
and $s \circ t \circ u = (s \circ t) \circ u \neq s \circ (t \circ u)$).
The operation $t^\circ$ that curries a term $t$ is defined inductively
by the equation
\[
(f(t_1,\dots,t_n))^\circ = f \circ t_1^\circ \circ \cdots \circ t_n^\circ
\]
The \emph{curried version of $\RR$} is given by $\RR^\circ = \{ \ell^\circ \to r^\circ \mid \ell \to r \in \RR \}$.
\end{definition}

For ground systems, currying reflects and preserves the normal form properties and confluence.
For reflection, a direct simulation argument works
($s \xr[\RR]{} t$ holds if and only if $s^\circ \xr[\RR^\circ]{} t^\circ$,
the image of ${-}^\circ$ is closed under rewriting by $\RR^\circ$,
and $s^\circ$ is an $\RR^\circ$-normal form if and only if $s$ is an $\RR$-normal form).
For preservation, Kenneway et al.~\cite{KKSV96} show that \UNR and \NFP
are preserved by currying for left-linear systems,
and that \UNC is preserved by currying for arbitrary TRSs.
Kahrs~\cite{K95} shows that currying preserves confluence of TRSs.
In the case that $\RR$ is finite,
currying can be performed in $\OO(\Vert\RR\Vert)$ time.
The resulting TRS is at most twice as large as the original TRS,
which follows from the inequality $|s^\circ| \leq 2|s|-1$
that can be shown by induction on $s$.

\begin{example}[continued from Example~\ref{running0}]
\label{running01}
The curried ground TRSs for $\UU$ and $\VV$ are
\begin{align*}
\UU^\circ &= \{ \m f \circ \m a \to \m a, \m f \circ \m a \to \m b, \m a \to \m a \}
\\
\VV^\circ &= \{ \m{a} \to \m{b}, \m{a} \to \m{f} \circ \m{a}, \m{b} \to \m{f} \circ (\m{f} \circ \m{b}), \m{f} \circ (\m{f} \circ (\m{f} \circ \m{b})) \to \m{b} \}
\end{align*}
For convenience, we will use the abbreviations
$\mfa = \m{f} \circ \m{a}$,
$\mfb = \m{f} \circ \m{b}$, 
$\mffb = \m{f} \circ \mfb$, and
$\mfffb = \m{f} \circ \mffb$
in later examples, so we may write
\(
\UU^\circ = \{ \mfa \to \m a, \mfa \to \m b, \m a \to \m a \}
\)
and
\(
\VV^\circ = \{ \m{a} \to \m{b}, \m{a} \to \mfa, \m{b} \to \mffb,
  \mfffb \to \m{b} \}
\).
\end{example}

%%%%%%%%%%%%%%%%%%%%%%%%%%%%%%%%%%%%%%%%%%%%%%%%%%%%%%%%%%%%%%%%%%%%%%%%%%%%%%
\subsection{Flattening}
\label{sec:flat}

\newcommand{\add}{\textbf{mk}^{[]}}

\begin{listing}
\raggedright\noindent
$\textbf{flatten}(\RR^\circ)$:
\begin{algorithmic}[1]
\STATE{$\EER \gets \varnothing$, $\RR^\flat \gets \varnothing$}
\FORALL{$\ell \to r \in \RR^\circ$}
\STATE{add $\add(\ell) \to \add(r)$ to $\RR^\flat$}
\COMMENT{$\EER$ is modified by $\add(\cdot)$}
\ENDFOR
\end{algorithmic}
\vskip1ex
$\add(f(s_1,\dots,s_n))$:
\begin{algorithmic}[1]
\FOR{$1 \leq i \leq n$}
\STATE{$c_i \gets \add(s_i)$}
\COMMENT{$c_i = [s_i]$}%
\ENDFOR
\IF{$f(c_1,\dots,c_n) \to c \in \EER$ for some $c$}
\RETURN{$c$}
\COMMENT{$c = [f(s_1,\dots,s_n)]$}%
\ELSE
\STATE{add $f(c_1,\dots,c_n) \to c$ to $\EER$, where $c$ is a fresh constant}
\COMMENT{$c = [f(s_1,\dots,s_n)]$}%
\RETURN{$c$}
\ENDIF
\end{algorithmic}
\caption{Computation of $\EE$ and $\RR^\flat$.}
\label{lst:flat}
\end{listing}
For efficient computations,
it is useful to represent the curried ground TRS $\RR^\circ$
using flat rules.
The idea of flattening goes back to Plaisted~\cite{P93}.
\begin{definition}
\label{def:flat}
Let a curried TRS $\RR^\circ$ be given.
Choose a fresh constant $[s]$ for each distinct subterm $s$ of $\RR^\circ$
($s \trianglelefteq \RR^\circ$),
and let $\FF^{[]} = \{[s] \mid s \trianglelefteq \RR^\circ \}$.
The \emph{flattening of $\RR^\circ$} is given by $(\RR^\flat, \EER)$ with $\EER$ and $\RR^\flat$ defined by
\begin{align*}
\EER &=
 \{ f([s_1], \dots, [s_n]) \to [f(s_1,\dots, s_n)] \mid f(s_1,\dots,s_n) \trianglelefteq \RR^\circ \}\\
\RR^\flat &= \{ [\ell] \to [r] \mid \ell \to r \in \RR^\circ \}
\end{align*}
\end{definition}
Note in the definition of $\EER$, $f$ is either a constant
(with $n = 0$) or $f = {\circ}$ (with $n = 2$).
Hence
$\EER$ contains a rule $c \to [c]$ for each constant subterm $c$ of $\RR^\circ$,
and a rule $[s_1] \circ [s_2] \to [s_1 \circ s_2]$ for the remaining subterms
$s_1 \circ s_2$ of $\RR^\circ$.
The sizes of $\EER$ and $\RR^\flat$ are both
$\OO(\Vert \RR \Vert)$
(which equals $\OO(\Vert \RR^\circ \Vert)$.
We can compute the systems $\EE$ and $\RR^\flat$
in $\OO(\Vert\RR\Vert\log\Vert\RR\Vert)$ time.
To this end, we may represent $\EER$
by a lookup table that maps left-hand sides of $\EER$ to
their corresponding right-hand sides,
and then employ the algorithm from Listing~\ref{lst:flat},
using an auxiliary function $\add(\cdot)$
that maps a ground term $s$ to $[s]$ while maintaining the
necessary rules in $\EER$ and creating fresh constants for
the subterms of $s$ as necessary.

\begin{example}[continued from Example~\ref{running01}]
\label{running1}
We introduce fresh constants
$[\m a]$, $[\m b]$, $[\m f]$,
$[\mfa]$,
$[\mfb]$,
$[\mffb]$, and
$[\mfffb]$.
The flattening of $\UU^\circ$ and $\VV^\circ$ results
in $(\UU^\flat,\EEU)$ and $(\VV^\flat, \EEV)$, respectively, where
\begin{align*}
\UU^\flat &= \{ [\mfa] \to [\m a], [\mfa] \to [\m b], [\m a] \to [\m a] \}
\\
\EEU &= \{ \m a \to [\m a], \m b \to [\m b], \m f \to [\m f],
 [\m f] \circ [\m a] \to [\mfa] \}
\\
\VV^\flat &= \{[\m{a}] \to [\mfa], [\m{a}] \to [\m{b}],
 [\m{b}] \to [\mffb], [\mfffb] \to [\m{b}]\}
\\
\EEV &= \EEU \cup \{
[\m{f}] \circ [\m{b}] \to [\mfb],
[\m{f}] \circ [\mfb] \to [\mffb],
[\m{f}] \circ [\mffb] \to [\mfffb]\}
\end{align*}
\end{example}

Both $\EER$ and $\EER^-$ are terminating and confluent
(in fact, orthogonal) systems.
Note that for all terms $s \trianglelefteq \RR^\circ$,
we have $s \xrn[\EER] [s]$ and $[s] \xrn[\EER^-] s$.
Therefore, we can reconstruct $\RR^\circ$ from $\EER$ and $\RR^\flat$,
and $\RR^\flat$ from $\EER$ and $\RR^\circ$,
as follows:
\begin{equation}
\label{eq:flat}
\RR^\circ = \{ (\ell \nf[\EER^-], r \nf[\EER^-]) \mid \ell \to r \in \RR^\flat \}
\qquad
\text{and}
\qquad
\RR^\flat = \{(\ell \nf[\EER], r \nf[\EER]) \mid \ell \to r \in \RR^\circ \}
\end{equation}
As a consequence of this observation we obtain the following proposition.
\begin{proposition}
\label{prop:flat}
\(
{\xln[\EER^-]} \cdot {\xr[\RR^\flat]{}} \cdot {\xrn[\EER^-]}
\subseteq {\xr[\RR^\circ]{}} \subseteq
{\xr[\EER]{*}} \cdot {\xr[\RR^\flat]{}} \cdot {\xl[\EER]{*}}
\).
\end{proposition}
\proof
From \eqref{eq:flat},
\(
{\xr[\RR^\circ]{}} \subseteq
{\xr[\EER]{*}} \cdot {\xr[\RR^\flat]{}} \cdot {\xl[\EER]{*}}
\)
is immediate.
To get the first subset relation,
note that the effect of reducing to normal form with respect to $\EER^-$
is to replace constant subterms $[s] \in \FF^{[]}$ by $s \in \TT(\FF^\circ)$;
hence we can compute $C[s] \nf[\EER^-]$ independently on $C$ and $s$.
In particular this applies to a rewrite step $C[\ell] \xr[\RR^\flat]{} C[r]$
with $\ell \to r \in \RR^\flat$.
Together with \eqref{eq:flat} it follows that
\(
{\xln[\EER^-]} \cdot {\xr[\RR^\flat]{}} \cdot {\xrn[\EER^-]}
\subseteq {\xr[\RR^\circ]{}}
\).
\qed
\begin{remark}
Confluence is actually preserved by flattening,
if one replaces $\RR^\circ$ by $\EER^\pm \cup \RR^\flat$
(this will be proved as part of Lemma~\ref{lem:rcr-cr}).
However, this is not the case for the normal form properties,
because the set of normal forms is not preserved.
Verma has given constructions that preserve $\UNR$ and $\UNC$,
but not confluence, in \cite{V09,VRL01}.
For ground systems
it is actually possible to preserve confluence and the normal form properties at the same time
if one replaces $\RR^\circ$ by $\RR' = \EER \cup \EER' \cup \RR^\flat$
where $\EER'$ is the subset of $\EER^-$ obtained
by only reversing those rules $[s_1] \circ [s_2] \to [s_1 \circ s_2] \in \EER$
for which one of the terms $s_1$ or $s_2$ is not a normal form.
The main points are that $s\nf[\EER]$ will be a normal form with respect $\RR'$ if and only if $s$ is a normal form with respect to $\RR$,
and that $s\nf[\EER] \xr[\RR']{*} t\nf[\EER]$ if and only if
$s \xr[\RR^\circ]{*} t$ for $s,t \in \TT(\FF^\circ)$.
We do not explore this idea here.
\end{remark}
\begin{remark}
An alternative view of the flattening step
that explains some of its utility for decision procedures
beyond restricting the shape of rules
is that it sets up a lookup table (namely, $\EER$)
for maximal sharing of the subterms of $\RR$.
This is precisely what the $\add$ function from
Listing~\ref{lst:flat} does.
\end{remark}

%%%%%%%%%%%%%%%%%%%%%%%%%%%%%%%%%%%%%%%%%%%%%%%%%%%%%%%%%%%%%%%%%%%%%%%%%%%%%%
\subsection{Recognizing Normal Forms}
\label{sec:nf}

The set of normal forms of a ground system is a regular language;
in fact this is true for left-linear systems~\cite{tata}.
We give a direct construction for ground TRSs.
We assume that we are given a curried ground TRS $\RR^\circ$
over a signature $\FF$.

With $Q^\flat = Q_f^\flat = \FF^{[]}$ and $\Delta^\flat = \EER$,
where $\FF^{[]}$ and $\EER$ are obtained by flattening
(Section~\ref{sec:flat}),
we obtain a deterministic tree automaton
that accepts precisely the subterms of $\RR^\circ$.
We modify this automaton to recognize normal forms instead.
\begin{definition}
Let $\star$ be a fresh constant
and let
\begin{align*}
Q = {}& Q_f = \FF^\NN = \{ [s] \mid \text{$s \trianglelefteq \RR^\circ$ and $s$ is $\RR^\circ$-normal form} \} \cup \{ [\star] \}
\\
\Delta = {}&\{ f([s_1],\dots,[s_n]) \to [f(s_1,\dots,s_n)] \mid [f(s_1,\dots,s_n)] \in Q_f - \{[\star]\} \} \cup {} \\
&\{ f([s_1],\dots,[s_n]) \to [\star] \mid \text{$f \in
\FF^\circ
$, $[s_1],\dots,[s_n] \in Q_f$, $f(s_1,\dots,s_n) \not\trianglelefteq \RR^\circ$} \}
\end{align*}
The automaton $\NN$ is given by $(Q,Q_f,\Delta)$.
Note that due to the restricted signature $\FF^\circ$,
all transitions have shape
$c \to [c]$, $[s_1] \circ [s_2] \to [s_1 \circ s_2]$, or
$[s_1] \circ [s_2] \to [\star]$,
where in the latter case, $s_1 = \star$ or $s_2 = \star$ may be used.
\end{definition}

\begin{example}[continued from Example~\ref{running1}]
\label{running-nf}
For $\NNU$, we have $Q = Q_f = \{[\m f],[\m b],[\star]\}$ and
$\Delta$ is given by
\begin{xalignat*}{6}
\m f &\to [\m f] &
[\m f] \circ [\m b] &\to [\star] &
[\m f] \circ [\m f] &\to [\star] &
[\m f] \circ [\star] &\to [\star] &
[\star] \circ [\m b] &\to [\star] &
[\star] \circ [\m f] &\to [\star] \\
\m b &\to [\m b] &
[\m b] \circ [\m b] &\to [\star] &
[\m b] \circ [\m f] &\to [\star] &
[\m b] \circ [\star] &\to [\star] &
[\star] \circ [\star] &\to [\star] &
\end{xalignat*}
For $\NNV$, we have $Q = Q_f = \{[\m{f}],[\star]\}$ and
$\Delta$ consists of the transitions
\begin{xalignat*}{5}
\m{f} &\to [\m f] &
[\m f] \circ [\m f] &\to [\star] &
[\m f] \circ [\star] &\to [\star] &
[\star] \circ [\m f] &\to [\star] &
[\star] \circ [\star] &\to [\star]
\end{xalignat*}
\end{example}

\begin{lemma}
The automaton $\NNR = (Q,Q_f,\Delta)$ recognizes the set of $\RR^\circ$-normal forms over $\FF^\circ$.
\end{lemma}
\begin{proof}
The idea is that the state $[\star]$ accepts those $\RR^\circ$-normal forms
$s$ that are not subterms of $\RR^\circ$,
while normal form subterms $s \trianglelefteq \RR^\circ$ are accepted in state $[s]$.
We show this by induction on $s$.

If $s = f(s_1,\dots,s_n)$ is accepted by $\NNR$ in state $q$,
then each $s_i$ is accepted by $\NNR$ as well,
which means that they are in normal form by the induction hypothesis.
If $q \in \FF^{[]}$,
then the last transition being used must be $f([s_1],\dots,[s_n]) \to [s]$,
with $s$ in normal form.
Otherwise, $q = [\star]$,
which means that either $s_i = \star$ for some $i$,
whence $s$ has a subterm that is not a subterm of $\RR^\circ$,
or $f(s_1,\dots,s_n) \not\trianglelefteq \RR^\circ$.
In either case, $s$ cannot be a left-hand side of $\RR^\circ$,
so it is in normal form.

If $s = f(s_1,\dots,s_n)$ is a normal form,
then each $s_i$ is also a normal form,
and hence accepted by $\NNR$ in a state $q_i$ by the induction hypothesis.
If any $s_i$ is accepted in state $[\star]$,
then $f(\dots,\star,\dots) \not\trianglelefteq \RR^\circ$ ensures
that $f(q_1,\dots,q_n) \to [\star]$ is a transition of $\NNR$.
Otherwise $q_i = [s_i]$ for all $i$,
and there is a transition $f(q_1,\dots,q_n) \to [s]$ or
$f(q_1,\dots,q_n) \to [\star]$ in $\NNR$
depending on whether $s \trianglelefteq \RR^\circ$.
In either case, $s$ is accepted by $\NNR$ as claimed.
\end{proof}
In order to compute $\NNR$ efficiently,
let $\overline Q_f = \{[s] \mid \text{$s \trianglelefteq \RR^\circ$ and $s$ is $\RR^\circ$-reducible}\}$,
which we may compute by the inference rules
\begin{gather*}
\inferr[side]{[\ell] \in \overline Q_f}{[\ell] \to [r] \in \RR^\flat}
\qquad
\inferr[arg$_i$]{[s_1 \circ s_2] \in \overline Q_f}{
[s_1] \circ [s_2] \to [s_1 \circ s_2] \in \EE &
[s_i] \in \overline Q_f
}
\end{gather*}
where $i \in \{1,2\}$.
The first rule (\iname{side}) expresses that left-hand sides of $\RR^\circ$ are reducible,
whereas the second rule (\iname{arg$_i$}) states that
if an argument of a term is reducible,
then the whole term is reducible as well;
because of the restricted signature of the curried system, the root symbol must be $\circ$ in this case.
This computation can be done in linear time
($\OO(\Vert \RR \Vert)$)
by Horn inference.
Once we have computed $\overline Q_f$, we can determine $Q = Q_f = (\FF^{[]} - \overline Q_f) \cup \{ [\star] \}$.
For $\Delta$, we can set up a partial function
\[
\delta\colon f([s_1],\dots,[s_n]) \mapsto q \quad\text{if}\quad
f([s_1],\dots,[s_n]) \to q \in \Delta
\]
where $[s_i] \in Q_f'$ for $1 \leq i \leq n$
that can be queried in logarithmic time, see Listing~\ref{lst:delta}.

\begin{listing}
\begin{algorithmic}[1]
\IF{there is a rule $f([s_1],\dots, [s_n]) \to [s] \in \EER$}
 \IF{$[s] \in Q_f$}
  \RETURN $[s]$ \COMMENT{normal form subterm of $
\RR^\circ
$}
 \ELSE
  \STATE no result \COMMENT{not a normal form}
 \ENDIF
\ELSE
 \RETURN $[\star]$ \COMMENT{normal form, not a subterm of $
\RR^\circ
$}
\ENDIF
\end{algorithmic}
\caption{Implementation of $\delta(f([s_1],\dots,[s_n]))$.}
\label{lst:delta}
\end{listing}

%%%%%%%%%%%%%%%%%%%%%%%%%%%%%%%%%%%%%%%%%%%%%%%%%%%%%%%%%%%%%%%%%%%%%%%%%%%%%%
\subsection{Congruence Closure}
\label{sec:cc}

Congruence closure (introduced by Nelson and Oppen~\cite{NO80}; a clean and fast implementation can be found in~\cite{NO07})
is an efficient method for deciding convertibility of ground terms modulo a set of ground equations; in our case, $\RR^\circ$.

The congruence closure procedure consists of two phases.
In the first phase, one determines the congruence classes (hence the name) among the subterms of the given set of equations,
where two subterms $s$ and $t$ of $\RR^\circ$ are identified
if and only if they are convertible, $s \xlr[\RR^\circ]{*} t$.
We write $[s]\EQU$ for the convertibility class
(also known as congruence class) of $s$ with respect to $\RR^\circ$.
The result of the first phase is a system of rewrite rules $\CCR$
that can be expressed concisely as
\[
\CCR =
\{ f([s_1]\EQU, \dots, [s_n]\EQU) \to [f(s_1,\dots,s_n)]\EQU \mid f(s_1,\dots,s_n) \trianglelefteq \RR^\circ \}
\]
The key point for efficiency is that it is enough to
consider subterms of $\RR^\circ$ when computing the congruence closure.
In an implementation,
a representative subterm of $\RR^\circ$ of $[s]\EQU$ will be stored instead of the whole class.
Let $\FFR = \{ [s]\EQU \mid s \trianglelefteq \RR^\circ \}$,
where we regard the congruence classes $[s]\EQU$ as fresh constants.
The size of $\CCR$ is $\OO(\Vert \RR \Vert)$.

\begin{remark}
The earlier notation $[s]$ for $s \trianglelefteq \RR^\circ$,
which suggests a congruence class, can be justified by viewing $[s]$ as
the congruence class of $s$ with respect to equality.
\end{remark}

In the second phase, given two terms $s$ and $t$,
one computes the normal forms with respect to rules in $\CCR$.
The terms $u$ and $v$ are $\RR^\circ$-convertible if and only if $u \nf[\CCR] = v \nf[\CCR]$.
We observe the following.

\begin{proposition}
With the signature $\FF^\circ \cup \FFR$,
the set $\CCR$ is an ortho\-gonal, terminating, ground TRS whose rules,
regarded as transitions of a tree automaton, are deterministic.
\end{proposition}

\proof
Termination of $\CCR$ is obvious.
Note that if
$f(s_1,\dots,s_n), f(t_1,\dots,t_n) \trianglelefteq \RR^\circ$
are subterms of $\RR^\circ$ such that
$f([s_1]\EQU,\dots,[s_n]\EQU) = f([t_1]\EQU,\dots,[t_n]\EQU)$,
then $s_i \xlr[\RR^\circ]{*} t_i$ for $1 \leq i \leq n$.
Consequently, $f(s_1,\dots,s_n) \xlr[\RR^\circ]{*} f(t_1,\dots,t_n)$,
so $[f(s_1,\dots,s_n)]\EQU = [f(t_1,\dots,t_n)]\EQU$.
Hence the rules are deterministic as claimed.
\qed

Consequently, we may represent $\CCR$ as a deterministic tree automaton
$\CCR = (Q, Q_f, \Delta)$ with $Q = Q_f = \FFR$ and $\Delta = \CCR$.
Each state $[s]\EQU$ accepts precisely the terms convertible to $s$.
Note that the automaton is not completely defined in general:
Only terms $s$ that allow a conversion $s \xlr[\RR^\circ]{*} t$ with a root step are accepted.%
\footnote{
The convertibility relation $\xlr[\RR^\circ]{*}$ is
accepted by the ground tree transducer $(\CCR, \CCR)$ (cf.~\cite{tata}).
}
The computation of $\CCR$ takes $\OO(\Vert\RR\Vert\log\Vert\RR\Vert)$ time~\cite{NO07},
and is based on flattening (Section~\ref{sec:flat}).
As a byproduct of the computation we obtain a map $(\cdot)\EQU : \FF^{[]} \to \FFR$
with $([s])\EQU = [s]\EQU$ for $s \trianglelefteq \RR^\circ$.
We lift this map to arbitrary terms over $\FF^\circ \cup \FF^{[]}$ by letting
$(f(s_1,\dots,s_n))\EQU = f((s_1)\EQU, \dots, (s_n)\EQU)$
for $f \in \FF^\circ$.
By this map, 
the key property of the congruence closure that convertible terms
over $\FF^\circ$
are joinable using rules in $\CCR$ can be strengthened
to incorporate flattening, as follows:

\begin{lemma}
\label{lem:cc}
For terms $s$ and $t$ over $\FF^\circ \cup \FF^{[]}$,
\begin{enumerate}
\item $s \xr[\EER]{} t$ implies $(s)\EQU \xr[\CCR]{} (t)\EQU$;
\item $s \xr[\RR^\flat]{} t$ implies $(s)\EQU = (t)\EQU$; and
\item $(s)\EQU \xr[\CCR]{} (t)\EQU$ implies $s\nf[\EER^-] \xlr[\RR^\circ]{*} t\nf[\EER^-]$.
\end{enumerate}
Consequently,
$s \xlr[\EER \cup \RR^\flat]{*} t$ if and only if
$(s)\EQU \xr[\CCR]{*} {\cdot} \xl[\CCR]{*} (t)\EQU$.
\end{lemma}
\begin{proof}
(1)~We have
$f([s_1]\EQU,\dots,[s_n]\EQU) \to [f(s_1,\dots,s_n)]\EQU \in \CCR$
for each
$f([s_1],\dots,[s_2]) \to [f(s_1,\dots,s_n)] \in \EER$,
so $s \xr[\EER]{} t$ implies $(s)\EQU \xr[\CCR]{} (t)\EQU$.
(2)~For $[\ell] \to [r] \in \RR^\flat$,
$\ell$ and $r$ are convertible by $\RR^\circ$,
and $[\ell]\EQU = [r]\EQU$ follows.
Consequently, $s \xr[\RR^\flat]{} t$ implies $(s)\EQU = (t)\EQU$.
(3)~Assume that $(s)\EQU \xr[\CCR]{=} (t)\EQU$.
We show $s\nf[\EER^-] \xlr[\RR^\circ]{*} t\nf[\EER^-]$ by induction on $s$.
If $(s)\EQU \xr[\CCR]{} (t)\EQU$ by a root step,
then $s$ and $t$ can be written as
$s = f([s_1],\dots,[s_n])$,
$t = [t']$,
and there is a subterm $f(s_1',\dots,s_n') \trianglelefteq \RR^\circ$ such that
\[
(s)\EQU = f([s_1]\EQU,\dots,[s_n]\EQU) =
 f([s_1']\EQU,\dots,[s_n']\EQU) \xr[\CCR]{} [f(s_1',\dots,s_n')]\EQU =
 [t']\EQU = (t)\EQU
\]
Recalling that $[s_i]\nf[\EER^-] = s_i$,
and that $[s_i]\EQU = [s_i']\EQU$ implies $s_i \xlr[\RR^\circ]{*} s_i'$,
we see that
\[
s\nf[\EER^-] =
f(s_1,\dots,s_n) \xlr[\RR^\circ]{*} f(s_1',\dots,s_n')
\xlr[\RR^\circ]{*} t' = t \nf[\EER^-]
\]
If $(s)\EQU \xr[\CCR]{=} (t)\EQU$ is not a root step, then either
$s = [s']$, $t = [t']$ with $[s']\EQU = [t']\EQU$,
and $s\nf[\EER^-] = s' \xlr[\RR^\circ]{*} t' = t\nf[\EER^-]$ follows,
or we can write $s = f(s_1,\dots,s_n)$, $t = f(t_1,\dots,t_n)$
such that $(s_i)\EQU \xr[\CCR]{=} (t_i)\EQU$ for $1 \leq i \leq n$.
We conclude by the induction hypothesis.

For the final claim, assume that $s \xlr[\EER \cup \RR^\flat]{*} t$.
Let $s' = s\nf[\EER^-]$ and note that $s' \xr[\EER]{*} s$,
and that $s' \in \TT(\FF^\circ)$.
Therefore, $s' = (s')\EQU \xr[\CCR]{*} (s)\EQU$.
Hence $s \xlr[\EER \cup \RR^\flat]{*} t$ implies
$s\nf[\EER^-] \xlr[\RR^\circ]{*} t\nf[\EER^-]$
by Proposition~\ref{prop:flat},
$(s)\EQU \xl[\CCR]{*} s\nf[\EER^-] \xlr[\CCR]{*} t\nf[\EER^-]\xr[\CCR]{*} (t)\EQU$
by the previous observation and the fact that $s\nf[\EER^-]$ and $t\nf[\EER^-]$ are terms over $\FF^\circ$,
and $(s)\EQU \xr[\CCR]{*} {\cdot} \xl[\CCR]{*} (t)\EQU$ by confluence of $\CCR$.
Conversely, $(s)\EQU \xr[\CCR]{*} {\cdot} \xl[\CCR]{*} (t)\EQU$
implies
$s \xl[\EER]{*} s\nf[\EER^-] \xlr[\RR^\circ]{*} t\nf[\EER^-] \xr[\EER]{*} t$
by part (3)
and
$s \xlr[\EER \cup \RR^\flat]{*} t$
by Proposition~\ref{prop:flat}.
\end{proof}

\begin{example}[continued from Example~\ref{running1}]
\label{running-cc}
For $\UU^\circ$,
restricted to the subterms of the system,
there are two congruence classes,
$[\m{f}]_\UU = \{\m{f}\}$ and
$[\m{a}]_\UU = \{\m{a},\m{b},\mfa\}$.
We have
\[
\CCU = \{
\m{f} \to [\m{f}]_\UU,
\m{a} \to [\m{a}]_\UU,
\m{b} \to [\m{a}]_\UU,
[\m{f}]_\UU \circ [\m{a}]_\UU \to [\m{a}]_\UU
\}
\]
For $\VV^\circ$,
restricted to the subterms of the system,
the congruence classes are $[\m{f}]_\VV = \{\m{f}\}$ and
$[\m{a}]_\VV = \{\m{a},\mfa,\m{b},\mfb,\mffb,\mfffb\}$,
and $\CCV$ is essentially the same as $\CCU$:
\[
\CCV = \{
\m{f} \to [\m{f}]_\VV,
\m{a} \to [\m{a}]_\VV,
\m{b} \to [\m{a}]_\VV,
[\m{f}]_\VV \circ [\m{a}]_\VV \to [\m{a}]_\VV
\}
\]
\end{example}

%%%%%%%%%%%%%%%%%%%%%%%%%%%%%%%%%%%%%%%%%%%%%%%%%%%%%%%%%%%%%%%%%%%%%%%%%%%%%%
\subsection{Rewrite Closure}
\label{sec:rc}

The rewrite closure is based on the flattened view
$(\RR^\flat,\EER)$ of the curried TRS $\RR^\circ$ (Section~\ref{sec:flat}).
In the following, $p$ and $q$ range over $\FF^{[]}$.
\begin{definition}
\label{def:rc}
The \emph{rewrite closure of $\RR^\circ$} is given by
\[
\textstyle
\RCR = \{ p \to q \mid p,q \in \FF^{[]} \text{ and } p \xr[\RR^\flat \cup \EE^\pm]{*} q \}
\]
\end{definition}
Note that whenever $p \to q \in \RCR$, $(p)\EQU = (q)\EQU$;
consequently, the map $(\cdot)\EQU$ is invariant under rewriting by $\RCR$.
\begin{proposition}
\label{prop:rcx}
$s \xr[\RCR]{} t$ implies $(s)\EQU = (t)\EQU$.
\qed
\end{proposition}

The point of the rewrite closure is that reachability in $\EER^\pm \cup \RR^\flat$
can be decomposed into a \emph{decreasing} sequence of steps in $\EER$ and $\RCR$,
followed by an \emph{increasing} sequence of steps in $\EER^-$ and $\RCR$.
Formally, we have the following lemma.

\begin{lemma}
\label{lem:rc}
$s \xr[\EER^\pm \cup \RR^\flat]{*} t$ if and only if
$s \xr[\EER \cup \RCR]{*} \cdot \xr[\EER^- \cup \RCR]{*} t$.
\end{lemma}

\begin{proof}
Consider Definition~\ref{def:rc}.
Since rules of $\RR^\flat$ have elements of $\FF^{[]}$ on both sides,
we have $\RR^\flat \subseteq \RCR$.
We also have $\RCR \subseteq {\xr[\EER^\pm \cup \RR^\flat]{*}}$.
Therefore, the reachability relation is preserved by the rewrite closure,
i.e.,
${\xr[\EER^\pm \cup \RR^\flat]{*}} = {\xr[\EER^\pm \cup \RCR]{*}}$.

Now consider a reduction $s \xr[\EER^\pm \cup \RR^\flat]{*} t$.
Then $s \xr[\EER^\pm \cup \RCR]{*} t$,
because of $\RR^\flat \subseteq \RCR$.
Assume that this rewrite
sequence is not of the shape
$s \xr[\EER \cup \RCR]{*} \cdot \xr[\EER^- \cup \RCR]{*} t$,
but has a minimal number of inversions between $\EER$ and $\EER^-$ steps
among the reductions $s \xr[\EER^\pm \cup \RCR]{*} t$
(an inversion is any pair of an $\EER$ step following an $\EER^-$ step,
not necessarily directly).
Then there is a subsequence of shape
$s' \xr[p,\EER^-]{} s'' \xr[\RCR]{*} t'' \xr[q,\EER]{} t'$, starting
with an $\EER^-$ step at position $p$ and a final $\EER$ step at position $q$.
The cases $p < q$ or $p > q$ are impossible,
because the rules of $\RCR$ only affect constants from $\FF^{[]}$,
and all rules from $\EER$ have a function symbol from $\FF^\circ$
at the root of their left-hand side, with constants from $\FF^{[]}$
as arguments.

If $p = q$ then
$s'|_p \xr[\EE^\pm \cup \RCR]{*} t'|_p$, which implies
$s'|_p \xr[\EE^\pm \cup \RR^\flat]{*} t'|_p$.
We also have $s'|_p, t'|_p \in \FF^{[]}$ because these two terms
are right-hand sides of rules in $\EE$.
Consequently, $s'|_p \to s'|_p \in \RCR$ follows by Definition~\ref{def:rc}.
Hence we can delete the two $\EER^\pm$ steps and the $\RCR$ steps
of $s'|_p \xr[\EE^\pm \cup \RCR]{*} t'|_p$,
and replace them by a single $\RCR$ step using the rule $s'|_p \to t'|_p$.
This decreases the number of inversions between $\EER$ and $\EER^-$ steps,
contradicting our minimality assumption.
Finally, if $p \parallel q$, then we
can reorder the rewrite sequence $s' \xr[\EER^\pm\cup\RCR]{*} t'$ as
$s' \xr[>q,\RCR]{*} \cdot \xr[q,\EER]{} \cdot \xr[\RCR]{*}\cdot\xr[p,\EER^-]{}\cdot\xr[>p,\RCR]{*} t'$,
commuting mutually parallel rewrite steps.
This reduces the number of inversions between $\EER$ and $\EER^-$ steps,
and again we reach a contradiction.
\end{proof}

The definition of $\RCR$ does not lend itself to an effective computation.
Lemma~\ref{lem:rc0} shows that we could alternatively define
$\RCR = {\leadsto}$,
where $\leadsto$ is determined by the inference rules in Figure~\ref{fig:rc}.
Note that the congruence rule (\iname{cong}) is specialized to the binary symbol $\circ$.
This is because all other elements $c \in \FF^\circ$ are constants,
and the congruence arising from $c \to [c] \in \EER$
is an instance of (\iname{refl}) using $p = [c]$.

\begin{lemma}
\label{lem:rc0}
We have $p \leadsto q$ if and only if $p \to q \in \RCR$.
\end{lemma}

\proof
By Definition~\ref{def:rc},
$p \to q \in \RCR$ is equivalent to $p \xr[\EER^\pm \cup \RCR]{*} q$.
All the inference rules in Figure~\ref{fig:rc} are consistent with
the requirement that $p \leadsto q$ implies $p \xr[\EER^\pm \cup \RCR]{*} q$.
The most interesting case is the (\irule{cong}) rule,
for which the following reduction from $p$ to $q$ is obtained:
\[
p \xr[\EER^-]{}
p_1 \circ p_2 \xr[\EER^\pm \cup \RCR]{*}
q_1 \circ p_2 \xr[\EER^\pm \cup \RCR]{*}
q_1 \circ q_2 \xr[\EER]{} q
\]

We show that $p \xr[\EER^\pm \cup \RCR]{*} q$ implies $p \leadsto q$ by contradiction.
Assume that $p \xr[\EER^\pm \cup \RCR]{*} q$ but not $p \leadsto q$.
Let $p = t_0 \to \dots \to t_n = q$ be a shortest sequence of
$(\EER^{\pm} \cup \RCR)$ steps from $p$ to $q$, and pick $p$ and $q$
such that $n$ is minimal. If $n = 0$ then $p = q$,
and $p \leadsto q$ by (\iname{refl}).
If $n = 1$ then $p \to q \in \RCR$ since $\EER$-rules
are rooted by elements of $\FF^\circ$ on their left-hand sides.
If $t_i \in \FF^{[]}$ for any $0 < i < n$, then
$p \leadsto t_i \leadsto q$ by minimality of
$p \xr[\EER^\pm \cup \RCR]{*} q$, and $p \leadsto q$ by transitivity
(\iname{trans}).
If any $t_i = c \in \FF$ then
we must have $t_{i-1} = [c] \xr[\EER^-]{} t_i \xr[\EER]{} [c] = t_{i+1}$,
and removing these two rewrite steps from the sequence
results in a shorter reduction from $p$ to $q$, contradicting minimality.
In the remaining case, we have $t_i = p_i \circ q_i$ for all
$0 < i < n$, and hence $p_1 \xr[\EER^\pm \cup \RCR]{*} p_{n-1}$
and $q_1 \xr[\EER^\pm \cup \RCR]{*} q_{n-1}$ since any root step
would have a constant from $\FF^{[]}$ as source or target. But
these two rewrite sequences have length at most $n-2$,
and therefore $p_1 \leadsto p_{n-1}$ and $q_1 \leadsto q_{n-1}$ hold.
This implies $p \leadsto q$ by the (\iname{cong}) rule.
In all cases we found that $p \leadsto q$, a contradiction.
\qed

\begin{figure}
\begin{gather*}
\inferr[refl]{p \leadsto p}{p \in \FF^{[]}}
\qquad
\inferr[base]{p \leadsto q}{p \to q \in \RR^\flat}
\qquad
\inferr[trans]{p \leadsto r}{p,q,r \in \FF^{[]} & p \leadsto q & q \leadsto r}
\\
\inferr[cong]{p \leadsto q}{p_1 \circ p_2 \to p \in \EER & p_1 \leadsto q_1 & p_2 \leadsto q_2 & q_1 \circ q_2 \to q \in \EER}
\end{gather*}
\caption{Inference rules for rewrite closure}
\label{fig:rc}
\end{figure}
\begin{listing}
\raggedright\noindent
$\textbf{rewrite-closure}(n, \EER, \RR^\flat)$: compute rewrite closure $\RCR = {\leadsto}$\\
\begin{algorithmic}[1]
\REQUIRE
$\FF^{[]} = \{ 1, \dots, n \}$, which can be achieved as part of the flattening step
\STATE $(\leadsto) \gets \varnothing \subseteq \{ 1, \dots, n \}^2$ \qquad (represented by an array).
\FOR{$1 \leq p \leq n$}
 \STATE $\textbf{add}(p,p)$ \COMMENT{(\iname{refl})}
\ENDFOR
\FORALL{$p \to q \in \RR^\flat$}
 \STATE $\textbf{add}(p,q)$ \COMMENT{(\iname{base})}
\ENDFOR
\end{algorithmic}

\vskip1ex

$\textbf{add}(p,q)$: add $p \leadsto q$ and process implied (\iname{trans}) and (\iname{cong}) rules\\
\begin{algorithmic}[1]
\IF {$p \leadsto q$}
 \RETURN
\ENDIF
\STATE extend $\leadsto$ by $p \leadsto q$
\FOR{$1 \leq r \leq n$}
 \IF{$r \leadsto p$}
  \STATE $\textbf{add}(r,q)$ \COMMENT{(\iname{trans})}
 \ENDIF
 \IF{$q \leadsto r$}
  \STATE $\textbf{add}(p,r)$ \COMMENT{(\iname{trans})}
 \ENDIF
\ENDFOR
\FORALL{$p \circ p_2 \to p_r \in \EER$ and $q \circ q_2 \to q_r \in \EER$ with $p_2 \leadsto q_2$}
 \STATE $\textbf{add}(p_r,q_r)$ \COMMENT{(\iname{cong})}
\ENDFOR
\FORALL{$p_1 \circ p \to p_r \in \EER$ and $q_1 \circ q \to q_r \in \EER$ with $p_1 \leadsto q_1$}
 \STATE $\textbf{add}(p_r,q_r)$ \COMMENT{(\iname{cong})}
\ENDFOR
\end{algorithmic}

\caption{Algorithm for rewrite closure}
\label{lst:rc-impl}
\end{listing}

\begin{example}[continued from Example~\ref{running1}]
\label{running-rc}
We present $\RCU$ and $\RCV$ as tables, where non-empty
entries correspond to the rules contained in each TRS,
and we leave out the surrounding brackets for the elements of $\FF^{[]}$.
For example, $[\m{a}] \to [\mfa] \notin \RCU$
but $[\m{a}] \to [\mfa] \in \RCV$.
The letters indicate the inference rule used to derive the entry, while
the superscripts indicate stage numbers---each inference uses only
premises that have smaller stage numbers.
\[
\RCU =
\begin{array}{r|cccc}
&\m{f}&\m{a}&\mfa&\m{b}\\
\hline
\m{f}& r^0 &     &     &     \\
\m{a}&     & r^0 &     &     \\
\mfa &     & b^0 & r^0 & b^0 \\
\m{b}&     &     &     & r^0 \\
\end{array}
\qquad
\RCV =
\begin{array}{r|cccccccc}
&\m{f}&\m{a}&\mfa&\m{b}&\mfb&\mffb&\mfffb\\
\hline
\m{f} & r^0 &     &     &     &     &     &     \\
\m{a} &     & r^0 & b^0 & b^0 & t^2 & t^1 & t^3 \\
\mfa  &     &     & r^0 & t^3 & c^1 & t^3 & t^2 \\
\m{b} &     &     &     & r^0 & t^4 & b^0 & t^4 \\
\mfb  &     &     &     & t^2 & r^0 & t^2 & c^1 \\
\mffb &     &     &     & t^4 & c^3 & r^0 & c^3 \\
\mfffb&     &     &     & b^0 & t^4 & t^1 & r^0 \\
\end{array}
\]
\end{example}

The size of $\leadsto$ (and hence $\RCR$) is bounded by $|\FF^{[]}|^2 \in \OO(\Vert\RR\Vert^2)$.
We can view the inference rules in Figure~\ref{fig:rc} as a system
of Horn clauses with atoms of the form $p \leadsto q$ ($p, q \in \FF^{[]}$).
This system can be solved in time proportional to the total size of the
clauses, finding a minimal solution for
the set $\RCR$.
There are
$|\FF^{[]}|$ instances of (\iname{refl}),
$|\RR^\flat| = |\RR|$ instances of (\iname{base}),
$|\FF^{[]}|^3$ instances of (\iname{trans})
and at most $|\FF^{[]}|^2$ instances of (\iname{cong}),
noting that $p_1$, $p_2$ are determined by $p$
and $q_1$, $q_2$ are determined by $q$.
Therefore, we can compute $\RCR$ in time $\OO(\Vert \RR \Vert^3)$.

\begin{remark}
In our implementation, we do not generate these Horn clauses explicitly.
Instead, whenever we make a new inference $p \to q \in \RCR$, we check all
possible rules that involve $p \to q \in \RCR$ as a premise. The result
is an incremental algorithm (see Listing~\ref{lst:rc-impl}).
From an abstract point of view, however,
this is essentially the same as solving the Horn clauses as stated above.
This remark also applies to inference rules presented later.
\end{remark}

%%%%%%%%%%%%%%%%%%%%%%%%%%%%%%%%%%%%%%%%%%%%%%%%%%%%%%%%%%%%%%%%%%%%%%%%%%%%%%
\section{Deciding \UNC}
\label{sec:unc}

In this section we develop an algorithm that decides \UNC
for a finite ground TRS $\RR$
in $\OO(\Vert\RR\Vert \log \Vert\RR\Vert)$ time for a given finite ground TRS $\RR$.
As preprocessing steps, we curry the TRS to bound its arity
while preserving and reflecting \UNC
(Section~\ref{sec:curry}),
compute the automaton $\NNR$ that accepts the $\RR^\circ$-normal forms
(Section~\ref{sec:nf}),
and the congruence closure $\CCR$
that allows for an efficient checking of convertibility (Section~\ref{sec:cc}).

First note that if we have two distinct convertible normal forms $s \xlr[\RR^\circ]{*} t$ such that the conversion does not contain a root step,
then there are strict subterms of $s$ and $t$ that are convertible and distinct.
Therefore, \UNC reduces to the question whether any state of $\CCR$,
the automaton produced by the congruence closure of $\RR^\circ$ (cf.~Section~\ref{sec:cc}),
accepts more than one normal form.
Let $\CCR \times \NNR$ be the result of the product construction on $\CCR$ and $\NNR$,
where $\NNR$ is the automaton accepting the $\RR^\circ$-normal forms (cf.~Section~\ref{sec:nf}).
We can decide $\UNC$ by enumerating accepting runs $t \xr[\CCR \times \NNR]{*} (p, q)$ in a bottom-up fashion (see Figure~\ref{fig:unc-enum}) until either
\begin{itemize}
\item
we obtain two distinct accepting runs ending in $(p_1,q_1)$ and $(p_2,q_2)$ with $p_1 = p_2$ (which means that the two corresponding terms are convertible normal forms),
in which case $\UNC(\RR)$ does not hold; or
\item
we have exhausted all runs,
in which case $\UNC(\RR)$ holds.
\end{itemize}
This enumeration of accepting runs is performed by the algorithm in Listing~\ref{lst:unc}.
Note that for achieving the desired complexity it is crucial that
on lines~\ref{c11} and~\ref{c13},
rather than iterating over all transitions of $\CCR$ and $\NNR$,
appropriate indices are used;
for line~\ref{c13}, this is the partial function $\delta$ from Section~\ref{sec:nf},
while for line~\ref{c11},
one can precompute an array that maps each state $p$ of $\CCR$ to
a list of transitions from $\CCR$
where $p$ occurs on the left-hand side.

\begin{example}[continued from Examples~\ref{running-nf} and~\ref{running-cc}]
\label{running-unc}
Consider $\UU$. The constant normal forms are $\m f$ and $\m b$,
so we add $([\m f]_\UU, [\m f], \m f)$ and $([\m b]_\UU = [\m a]_\UU, [\m b], \m b)$
to the \textit{worklist}. The main loop (lines~\ref{c5}--\ref{c18}),
is entered three times.
\begin{enumerate}
\item
$([\m f]_\UU, [\m f], \m f)$ is taken from the \textit{worklist},
and we assign $\textit{seen}([\m f]_\UU) = ([\m f], \m f)$.
There are no transitions in $\CCU$ satisfying the conditions in lines~\ref{c11} and \ref{c12}.
\item
$([\m a]_\UU, [\m b], \m b)$ is taken from the \textit{worklist},
and we assign $\textit{seen}([\m a]_\UU) = ([\m b],\m b)$.
This time, the transitions
$[\m f]_\UU \circ [\m a]_\UU \to [\m a]_\UU \in \CCU$
and $[\m f] \circ [\m b] \to [\star] \in \NNR$ satisfy the
conditions in lines \ref{c11}--\ref{c13},
and we push $([\m a]_\UU, [\star], \m f \circ \m b)$ to the \textit{worklist}.
\item
$([\m a]_\UU, [\star], \m f \circ \m b)$ is taken from the \textit{worklist}.
Since $\textit{seen}([\m a]_\UU)$ is already defined,
we conclude that $\UNC(\UU)$ does not hold.
Indeed $\m b$
(from $\textit{seen}([\m a]_\UU) = ([\m b],\m b)$)
and $\m f \circ \m b$ are convertible normal forms for $\UU^\circ$:
$\m b \xl{} \m f \circ \m a \xl{} \m f \circ (\m f \circ \m a) \xr{} \m f \circ \m b$.
\end{enumerate}

For $\VV$, there is only one constant normal form, namely $\m f$.
Hence, initially, we add $([\m f]_\VV, [\m f], \m f)$ to the \textit{worklist}.
Then we enter the main loop (lines~\ref{c5}--\ref{c18}).
On line~\ref{c10}, $\textit{seen}([\m f])$ is set to $([\m f]_\VV,\m f)$.
At this point,
the only way that the check on line~\ref{c12} could succeed
would be having a transition with left-hand side $[\m f]_\VV \circ [\m f]_\VV$
in $\CCV$, but there is no such a transition.
Hence the main loop terminates and we conclude that $\UNC(\VV)$ is true.
\end{example}

\begin{figure}
\begin{gather*}
\inferr[const]{c \xr[\CCR \times \NNR]{*} ([c]\EQU,[c])}{c \trianglelefteq \RR^\circ & \text{$c$ normal form}} \\
\inferr[app]{t_1 \circ t_2 \xr[\CCR \times \NNR]{*} (p,q)}
{t_1 \xr[\CCR \times \NNR]{*} (p_1,q_1) &
 t_2 \xr[\CCR \times \NNR]{*} (p_2,q_2) &
 p_1 \circ p_2 \xr[\CCR]{} p &
 q_1 \circ q_2 \xr[\NNR]{} q
}
\end{gather*}
\caption{Enumerating runs of $\CCR \times \NNR$.}
\label{fig:unc-enum}
\end{figure}
\begin{listing}
\begin{algorithmic}[1]
\STATE compute $\CCR$ and a representation of $\NNR$ \label{c1}
\STATE let $\mathit{seen}(p)$ be undefined for all $p \in \FFR$ (to be updated below)
\FORALL{constants $c \trianglelefteq \RR^\circ$ that are normal forms} \label{c2}
 \STATE push $([c]\EQU,[c], c)$ to \textit{worklist} \label{c3}
  \COMMENT{(\iname{const})}
\ENDFOR \label{c4}
\WHILE{\textit{worklist} not empty} \label{c5}
 \STATE $(p,q,s) \gets$ pop \textit{worklist} \label{c6}
\COMMENT{$s \xr[\CCR \times \NNR]{*} (p,q)$}
 \IF{$\mathit{seen}(p)$ is defined} \label{c7}
  \RETURN $\UNC(\RR)$ is false \label{c8}
  \COMMENT{distinct convertible normal forms}
 \ENDIF \label{c9}
 \STATE $\mathit{seen}(p) \gets (q,s)$ \label{c10}
 \FORALL{transitions $p_1 \circ p_2 \to p_r \in \CCR$ with $p \in \{p_1,p_2\}$} \label{c11}
  \IF{$\mathit{seen}(p_1) = (q_1,s_1)$ and $\mathit{seen}(p_2) = (q_2,s_2)$ are defined} \label{c12}
   \IF{there is a transition $q_1 \circ q_2 \to q_r \in \NNR$} \label{c13}
    \STATE push $(p_r,q_r,s_1 \circ s_2)$ to \textit{worklist} \label{c14}
    \COMMENT{(\iname{app})}
   \ENDIF \label{c15}
  \ENDIF \label{c16}
 \ENDFOR \label{c17}
\ENDWHILE \label{c18}
\RETURN $\UNC(\RR)$ is true \label{c19}
\COMMENT{runs exhausted}
\end{algorithmic}
\caption{Deciding $\UNC(\RR)$}
\label{lst:unc}
\end{listing}

Correctness of the procedure hinges on two key facts:
First, the automaton $\CCR \times \NNR$ is deterministic,
which means that distinct runs result from distinct terms.
Secondly, the set of $\RR^\circ$-normal forms is closed under taking subterms,
so we can skip non-normal forms in the enumeration.

\begin{lemma}
\label{lem:unc-inv}
Whenever the algorithm in Listing~\ref{lst:unc} reaches line~\ref{c5},
the following conditions are satisfied:
\begin{enumerate}
\item
$seen(p) = (q,s)$ or $(p,q,s) \in \mathit{worklist}$ implies a run $s \xr[\CCR \times \NNR]{*} (p,q)$;
\item
the elements of $\mathit{worklist}$ are distinct,
and
$seen(p) = (q,s)$ implies $(p,q,s) \notin \mathit{worklist}$; and
\item
for any run $s \xr[\CCR \times \NNR]{*} (p,q)$, either
$seen(p) = (q,s)$, or
there is a subterm $s' \trianglelefteq s$
and a run $s' \xr[\CCR \times \NNR]{*} (p',q')$
such that $(p',q',s') \in \mathit{worklist}$.
\end{enumerate}
\end{lemma}
\proof
For the first property,
note that only lines~\ref{c3} and~\ref{c14} extend the \textit{worklist}.
On line~\ref{c3}, we have $c \xr[\CCR \times \NNR]{*} ([c]\EQU,[c])$ by (\iname{const})
from Figure~\ref{fig:unc-enum}.
On line~\ref{c14},
we have $s_1 \circ s_2 \xr[\CCR \times \NNR]{*} (p_r,q_r)$
by (\iname{app}).
Further note that when we update $\textit{seen}(p)$ on line~\ref{c10},
$s \xr[\CCR \times \NNR]{*} (p,q)$ holds because $(p,q,s)$ was previously on the $\mathit{worklist}$.

For the second property, first note that $s$ determines the values of $p$ and $q$ in all relevant triples $(p,q,s)$,
so we can focus on the terms produced in lines~\ref{c3} and~\ref{c14}.
Line~\ref{c3} produces all possible constants for $s$, each exactly once.
As long as the procedure doesn't terminate,
any term removed from the \textit{worklist} on line~\ref{c6} ends up in the $\mathit{seen}$ map on line~\ref{c10},
so as long as the invariant is maintained, the terms $s$ obtained on line~\ref{c6} will be distinct.
Now a term $s_1 \circ s_2$ will only be produced on line~\ref{c14} when
the term $s$ just taken from the \textit{worklist} equals $s_1$ or $s_2$,
and furthermore, either $s_1 = s_2$, or the other term had previously
been taken from the \textit{worklist}.
Clearly this can happen only once, establishing the invariant.

For the third invariant,
let us say that the invariant holds for $s$ if
there is a run $s \xr[\CCR \times \NNR]{*} (p,q)$ and the
third invariant holds for that run.
First note that every run $s \xr[\CCR \times \NNR]{*} (p,q)$ must use (\iname{const})
at the leafs,
so any constant subterm $c \trianglelefteq s$ will
satisfy the invariant after the initialization on lines \ref{c1}--\ref{c4}
is completed.
We show by induction on $s$
that the invariant for $s$ is maintained in the main loop.
Note that once $\mathit{seen}(p) = (q,s)$,
that will remain true because line~\ref{c10} never updates a defined value of $\mathit{seen}$.
Furthermore, if $(p,q,s) \in \textit{worklist}$,
then the invariant will be maintained until $(p,q,s)$ is taken from the \textit{worklist} on line~\ref{c6},
but then the invariant will be restored on line~\ref{c10}.
In particular, the invariant is maintained for constants $s$.
So assume that $s$ is not constant, i.e., $s = s_1 \circ s_2$,
and that neither $\textit{seen}(p) = (q,s)$ nor $(p,q,s) \in \textit{worklist}$.
We can decompose the run $s \xr[\CCR \times \NNR]{*} (p,q)$ as
$s_1 \circ s_2 \xr[\CCR \times \NNR]{*} (p_1,q_1) \circ (p_2,q_2) \xr[\CCR \times \NNR]{} (p,q)$.
At the end of the loop body (line~\ref{c17}),
the invariant will hold for $s_1 \xr[\CCR \times \NNR]{*} (p_1,q_1)$ and
$s_2 \xr[\CCR \times \NNR]{*} (p_2,q_2)$.
This ensures the invariant for $s$ except when
$\textit{seen}(p_1) = (q_1,s_1)$ and $\textit{seen}(p_2) = (q_2,s_2)$.
But in that case,
because the invariant held for $s \xr[\CCR \times \NNR]{*} (p,q)$ at the beginning of the loop iteration,
one of $\textit{seen}(p_1)$ or $\textit{seen}(p_2)$ must just have been assigned,
i.e., either $(p_1,q_1,s_1)$ or $(p_2,q_2,s_2)$ must have been taken from the \textit{worklist} on line~\ref{c6}.
Hence the transitions $p_1 \circ p_2 \to p \in \CCR$ and $q_1 \circ q_2 \to q \in \NNR$ satisfy the conditions on lines~\ref{c12}--\ref{c14}
and $(p,q,s)$ is added to the \textit{worklist} on line~\ref{c15},
restoring the invariant for $s$.
\qed

\begin{theorem}
\label{thm:unc}
The algorithm in Listing~\ref{lst:unc} is correct and runs in $\OO(\Vert\RR\Vert \log \Vert\RR\Vert)$ time.
\end{theorem}
\begin{proof}
The first two invariants established in Lemma~\ref{lem:unc-inv} ensure
that if $\mathit{seen}(p) = (q',s')$ on line~\ref{c8},
then we have distinct runs $s \xr[\CCR \times \NNR]{*} (p,q)$
and $s' \xr[\CCR \times \NNR]{*} (p,q')$ witnessing
that $s$ and $s'$ are distinct convertible normal forms, so \UNC does not hold.
On the other hand, if line~\ref{c19} is reached,
then by the third invariant,
we have $seen(p) = (q,s)$
for all runs $s \xr[\CCR \times \NNR]{*} (p,q)$,
showing that any normal form in the convertibility class represented by $p$
equals $s$, which establishes \UNC.
Termination follows from the complexity bound.

So let us focus on the complexity.
Let $n = \Vert\RR\Vert$.
As stated in Section~\ref{sec:cc},
computing $\CCR$ takes $\OO(n \log n)$ time.
In Section~\ref{sec:nf} we presented a $\OO(n \log n)$ time computation
for $\NNR$ where the transitions are represented
by a partial function $\delta$ that can be queried in $\OO(\log n)$ time
per left-hand side.
This covers line~\ref{c1} of the algorithm.
Lines~\ref{c2} and~\ref{c4} take $\OO(n)$ time
(the normal form constants are already enumerated in the computation of $\NNR$).
Observe that lines \ref{c10}--\ref{c17} update $\mathit{seen}$,
so they are executed at most once per element of $\FFR$, i.e., $\OO(n)$ times.
The precomputation for line~\ref{c11} takes $\OO(n)$ time as well.
In line~\ref{c11}, each transition $p_1 \circ p_2 \to p_r \in \CCR$
is encountered at most twice, once for $p_1$
and a second time for $p_2$,
which means that lines \ref{c12}--\ref{c16}
are executed at most twice per transition in $\CCR$, so $\OO(n)$ times.
The check on line~\ref{c13} takes $\OO(\log n)$ time per iteration, so $\OO(n \log n)$ time in total.
Finally, we note that line~\ref{c14} is executed $\OO(n)$ times, so no more than $\OO(n)$ items are ever added to the \textit{worklist},
which means that lines \ref{c5}--\ref{c9} are executed $\OO(n)$ times.
Overall the algorithm executes in $\OO(n \log n)$ time, as claimed.
\end{proof}
\begin{remark}
The third component of the elements of $\mathit{worklist}$
as well as the second component of the results of $\mathit{seen}$
in Listing~\ref{lst:unc}
were added to aid the proof of Theorem~\ref{thm:unc}
and can be omitted in an actual implementation.
\end{remark}

%%%%%%%%%%%%%%%%%%%%%%%%%%%%%%%%%%%%%%%%%%%%%%%%%%%%%%%%%%%%%%%%%%%%%%%%%%%%%%
\section{Deciding \UNR}
\label{sec:unr}

In this section we present a procedure that decides \UNR
for a finite ground TRS $\RR$
in $\OO(\Vert\RR\Vert^3 \log \Vert\RR\Vert)$ time.
As preprocessing, we curry the TRS to bound its arity
while preserving and reflecting \UNR
(Section~\ref{sec:curry}),
compute the automaton $\NNR$ that accepts the $\RR^\circ$-normal forms
(Section~\ref{sec:nf}),
and the rewrite closure $(\RCR, \EER)$ that
allows checking reachability (Section~\ref{sec:rc}).

% % % % % % % % % % % % % % % % % % % % % % % % % % % % % % % % % % % % % % % 
\subsection{Peak Analysis}

In order to derive conditions for \UNR,
assume that \UNR does not hold and consider
a peak $s \xl[\RR^\circ]{!} {\cdot} \xr[\RR^\circ]{!} t$
between distinct $\RR^\circ$-normal forms $s$ and $t$
of minimum total size.
If the peak has no root step then we can project it to the arguments,
and obtain a smaller counterexample.
So, without loss of generality,
we may assume that the peak has a root step in its left part.
Using the rewrite closure (cf.\ Lemma~\ref{lem:rc}),
the peak can be decomposed as
\[
s \xr[\EER \cup \RCR^-]{*} {\cdot} \xl[\EER \cup \RCR]{*} {\cdot} \xr[\EER \cup \RCR]{*} {\cdot} \xl[\EER \cup \RCR^-]{*} t
\]
where, in fact, the left part has a root step.
This means that there is a constant $p$ such that
\begin{equation}
\label{peak-2}
s \xr[\EER \cup \RCR^-]{*} p \xl[\EER \cup \RCR]{*} {\cdot} \xr[\EER \cup \RCR]{*} {\cdot} \xl[\EER \cup \RCR^-]{*} t
\end{equation}
First we consider the special case
\(
s \xr[\EER \cup \RCR^-]{*} p \xl[\EER \cup \RCR^-]{*} t
\),
which implies that any $p$ is reachable from at most one normal form
using rules from $\EER \cup \RCR^-$.
\begin{definition}
\label{def:unr-1}
The \emph{first \UNR-condition} holds if
\(
s \xr[\EER \cup \RCR^-]{*} p \xl[\EER \cup \RCR^-]{*} t
\)
for $\RR^\circ$-normal forms $s, t \in \TT(\FF^\circ)$
and $p \in \FF^{[]}$
implies $s = t$.
\end{definition}
If the first \UNR-condition is violated, then \UNR clearly does not hold.
Assume that the first \UNR-condition is satisfied and
let $w$ be the partial function that witnesses this fact:
\[
w(p) = t \iff \text{$t$ is a normal form and $\textstyle t \xr[\EER \cup \RCR^-]{*} p$}
\]
In order to analyze the conversion~\eqref{peak-2} further,
note that every forward $\EE$ step between $p$ and $t$ is preceded
by an inverse $\EE$ step at the same position,
inducing a peak between two constants from $\FF^{[]}$.
This situation is captured by the following definition.
\begin{definition}
\label{def:mc}
Two constants $p,q \in \FF^{[]}$ are \emph{meetable}
if $p \xl[\EER \cup \RCR]{*} {\cdot} \xr[\EER \cup \RCR]{*} q$.
In this case we write $p \uparrow q$.
\end{definition}
Using $\uparrow^\mnshortparallel$,
the parallel closure of the relation $\uparrow$,
we can find a multi-hole context $C$ and
constants $q_1, \dots, q_n$, $p_1,\dots,p_n$
such that~\eqref{peak-2} becomes
\begin{equation}
\label{peak-4}
s \xr[\EER \cup \RCR^-]{*} q \xl[\EER \cup \RCR]{*} C[q_1,\dots,q_n] \uparrow^\mnshortparallel C[p_1,\dots,p_n] \xl[\EER \cup \RCR^-]{*} t
\end{equation}
Note that because $s$ and $t$ are normal forms,
we have $s = w(q)$ and $t = C[w(p_1),\dots,w(p_n)]$.
\begin{definition}
\label{def:unr-2}
The \emph{second \UNR-condition} holds if
\eqref{peak-4} for $\RR^\circ$-normal forms $s, t \in \TT(\FF^\circ)$
implies $s = t$.
\end{definition}
The analysis of this section is summarized by the following lemma.
\begin{lemma}
\label{lem-unr}
A curried ground TRS is \UNR if and only if the first and second
\UNR-conditions are satisfied.
\qed
\end{lemma}

% % % % % % % % % % % % % % % % % % % % % % % % % % % % % % % % % % % % % % % 
\subsection{Computing Meetable Constants}
\label{sec:mc}

The meetable constant relation $\uparrow$ can be computed in a
way similar to the rewrite closure from Section~\ref{sec:rc},
using the inference rules in Figure~\ref{fig:mc}.
So for this subsection, let $\uparrow$ be defined by those inference rules.
The following lemma shows that
$\uparrow$ coincides with the meetable constants relation,
justifying the symbol.

\begin{lemma}
\label{lem:mc}
For $p,q \in \FF^{[]}$,
$p \uparrow q$
if and only if $p \xl[\EER \cup \RCR]{*} {\cdot} \xr[\EER \cup \RCR]{*} q$.
\end{lemma}

\proof
The proof follows the same principles as that of Lemma~\ref{lem:rc0},
so let us be brief.
First note that all rules in Figure~\ref{fig:mc} are consistent
with the requirement that $p \uparrow q$ implies $p \xl[\EER \cup \RCR]{*} {\cdot} \xr[\EER \cup \RCR]{*} q$.

On the other hand,
assume that there is peak
$p \xl[\EER \cup \RCR]{*} {\cdot} \xr[\EER \cup \RCR]{*} q$
such that $p \uparrow q$ does not hold.
Choose such a peak of minimal length.
Then either $p = q$, and (\irule{refl}) applies,
or $p \xl[\RCR]{} p' \uparrow q$ and (\irule{step$_l$}) applies,
or $p \uparrow q' \xr[\RCR]{} q$ and (\irule{step$_r$}) applies,
or $p \xl[\EE]{} p_1 \circ p_2 \uparrow^\mnshortparallel q_1 \circ q_2 \xr[\EE]{} q$,
in which case (\irule{cong}) applies.
In each case, $p \uparrow q$ follows, contradicting the assumption.
\qed

\begin{figure}
\begin{gather*}
\inferr[refl]{p \uparrow p\mathstrut}{p \in \FF^{[]}}
\qquad
\inferr[cong]{p \uparrow q}{p_1 \circ p_2 \to p \in \EER & p_1 \uparrow q_1 & p_2 \uparrow q_2 & q_1 \circ q_2 \to q \in \EER}
\\
\inferr[step\/$_l$]{p \uparrow r}{q \to p \in \RCR & q \uparrow r & r \in \FF^{[]}}
\qquad
\inferr[step\/$_r$]{p \uparrow r}{p \in \FF^{[]} & p \uparrow q & q \to r \in \RCR}
\end{gather*}
\caption{Inference rules for meetable constants}
\label{fig:mc}
\end{figure}
Note that as in the case of the rewrite closure,
there is no (\iname{cong}) rule for constants $c \in \FF^\circ$,
because they would be instances of (\iname{refl}).
There are $\OO(\Vert \RR \Vert)$ instances of (\iname{refl}),
There are $\OO(\Vert \RR \Vert^2)$ instances of (\iname{cong})
(because $p$ and $q$ determine $p_1$, $p_2$, $q_1$, and $q_2$),
and $\OO(\Vert \RR \Vert^3)$ instances each of (\iname{step}$_l$)
and (\iname{step}$_r$).
Using Horn inference with $p \uparrow q$ as atoms,
the computation of $\uparrow$ takes $\OO(\Vert\RR\Vert^3)$ time.

\begin{example}[continued from Example~\ref{running-rc}]
\label{running-mc}
In the same spirit as Example~\ref{running-rc},
we present the relations $\uparrow_\UU$ and $\uparrow_\VV$ as tables,
with the entries indicating rules and stage numbers.
\[
\uparrow_\UU =
\begin{array}{r|ccccc}
&\m{f}&\m{a}&\mfa&\m{b}\\
\hline
\m{f}& r^0 &     &       &       \\
\m{a}&     & r^0 & s^1_l & s^2_r \\
\mfa &     & s^1_r & r^0 & s^1_r \\
\m{b}&     & s^2_l & s^1_l & r^0
\end{array}
\qquad
\uparrow_\VV =
\begin{array}{r|cccccccc}
&\m{f}&\m{a}&\mfa&\m{b}&\mfb&\mffb&\mfffb\\
\hline
\m{f} & r^0 &     &     &     &     &     &     \\
\m{a} &     & r^0 & s^1_l & s^1_l & s^1_l & s^1_l & s^1_l \\
\mfa  &     & s^1_r & r^0 & s^1_l & s^1_l & s^1_l & s^1_l \\
\m{b} &     & s^1_r & s^1_r & r^0 & s^1_l & s^1_l & s^1_l \\
\mfb  &     & s^1_r & s^1_r & s^1_r & r^0 & s^1_l & s^1_l \\
\mffb &     & s^1_r & s^1_r & s^1_r & s^1_r & r^0 & s^1_l \\
\mfffb&     & s^1_r & s^1_r & s^1_r & s^1_r & s^1_r & r^0 \\
\end{array}
\]
\end{example}

% % % % % % % % % % % % % % % % % % % % % % % % % % % % % % % % % % % % % % % 
\subsection{Checking \UNR}
\begin{figure}
\begin{gather*}
\inferr[base]{c \in W([c],[c])}{c \trianglelefteq \RR^\circ & \text{$c$ normal form}}
\qquad
\inferr[step$_\RCR$]{s \in W(p',q)}{s \in W(p,q) & p' \to p \in \RCR}
\\
\inferr[step$_\EER$]{s_1 \circ s_2 \in W(p,q)}{
s_1 \in W(p_1,q_1) &
s_2 \in W(p_2,q_2) &
p_1 \circ p_2 \to p \in \EER &
q_1 \circ q_2 \to q \in \NNR}
\end{gather*}
\caption{Inference rules for $s \in W(p,q)$}
\label{fig:unr-1}
\end{figure}
\begin{listing}
\begin{algorithmic}[1]
\STATE compute rewrite closure $(\RCR, \EER)$ and a representation of $\NNR$ \label{d1}
\STATE let $w(p)$ and $n(p)$ be undefined for all $p \in \FF^{[]}$ (to be updated below)
\FORALL{constants $c \trianglelefteq \RR^\circ$ that are normal forms} \label{d2}
 \STATE push $([c],[c], c)$ to \textit{worklist} \label{d3}
\COMMENT{(\irule{base})}
\ENDFOR \label{d4}
\WHILE{\textit{worklist} not empty} \label{d5}
 \STATE $(p,q,s) \gets$ pop \textit{worklist} \label{d6}
\COMMENT{$s \in W(p,q)$}
 \IF[$t \in W(p,n(p))$]{$w(p) = t$ is defined} \label{d7}
  \IF{$s \neq t$} \label{d8}
   \RETURN $\UNR(\RR)$ is false \label{d9}
\COMMENT{first \UNR-condition violated by $s$ and $t$}
  \ENDIF \label{d10}
  \STATE \textbf{continue} at~\ref{d5}\label{d11}
 \ENDIF \label{d12}
 \STATE $w(p) \gets s$, $n(p) \gets q$ \label{d13}
 \FORALL{rules $p_1 \circ p_2 \to p_r \in \EER$ with $p \in \{p_1,p_2\}$} \label{d14}
  \IF{$w(p_1) = s_1$ and $w(p_2) = s_2$ are defined} \label{d15}
   \IF{there is a transition $n(p_1) \circ n(p_2) \to q_r \in \NNR$} \label{d16}
    \STATE push $(p_r,q_r,s_1 \circ s_2)$ to \textit{worklist} \label{d17}
\COMMENT{(\irule{step$_\EER$})}
   \ENDIF \label{d18}
  \ENDIF \label{d19}
 \ENDFOR \label{d20}
 \FORALL{rules $p' \to p \in \RCR$} \label{d21}
  \STATE push $(p',q,s)$ to \textit{worklist} \label{d22}
\COMMENT{(\irule{step$_\RCR$})}
 \ENDFOR \label{d23}
\ENDWHILE \label{d24}
\end{algorithmic}
\caption{Checking the first \UNR-condition and computing $w$ and $n$.}
\label{lst:unr-1}
\end{listing}

We start by checking the first \UNR-condition.
To perform this computation efficiently,
we make use of the automaton $\NNR$ that recognizes normal forms,
cf.~Section~\ref{sec:nf}.
The fact that $s$ is a $\RR^\circ$-normal form is witnessed by a run $s \xr[\NNR]{*} q$.
Let
\[
\textstyle
W(p,q) = \{ s \mid
s \in \TT(\FF^\circ) \text{, }
s \xr[\EER \cup \RCR^-]{*} p \in \FF^{[]} \text{ and }
s \xr[\NNR]{*} q \in \FF^\NN
\}
\]
\begin{lemma}
\label{lem:w}
The predicate $s \in W(p,q)$ is characterized by the inference rules in Figure~\ref{fig:unr-1}.
\end{lemma}
\begin{proof}
The inference rules follow by an inductive analysis of the last step
of the $s \xr[\EER \cup \RCR^-]{*} p$ reduction,
where $s \xr[\NNR]{*} q$.
Recall that $\NNR$ is deterministic, so $q$ is determined by $s$.
\begin{itemize}[align=left]
\item[(\irule{base})]
If there is a single step, then it must be using a rule $s = c \to [c] = p$ from $\EER$,
where $c \in \FF$, and $c$ is a $\RR^\circ$-normal form,
which ensures that $c \to [c] = q \in \NNR$ as well.
\item[(\irule{step$_\RCR$})]
If the last step is an $\RCR^-$ step,
then $s \xr[\EER \cup \RCR^-]{*} p' \xr[\RCR^-]{} p$,
and there is a $q$ with $s \in W(p,q)$.
\item[(\irule{step$_\EER$})]
If the last step is an $\EER$ step but $s$ is not a constant,
then $s = s_1 \circ s_2 \xr[\EER \cup \RCR^-]{*} p_1 \circ p_2 \xr[\EER]{} p$,
and there are $q_1$, $q_2$ with $s_1 \in W(p_1,q_1)$ and $s_2 \in W(p_2,q_2)$.
\end{itemize}
Conversely, each derivation of $s \in W(p,q)$ by these three
inference rules gives rise to rewrite sequences
$s \xr[\EER \cup \RCR^-]{*} p$ and $s \xr[\NNR]{*} q$.
\end{proof}
The corresponding code is given in Listing~\ref{lst:unr-1}.
In addition to $w(q)$ (which we introduced immediately after
Definition~\ref{def:unr-1})
we also compute a partial function $n(q)$
which returns the state of $\NNR$ that accepts $w(q)$
if the latter is defined.
The computed witnesses may have exponential size
(see Example~\ref{ex:exp-unr1}),
so in order to make the check on line~\ref{d8} efficient,
it is crucial to use maximal sharing.

\begin{example}[continued from Example~\ref{running-nf}]
\label{running-unr1}
Let us check the first \UNR-condition for
$\UU^\circ = \{ \mfa \to \m a, \mfa \to \m b, \m a \to \m a \}$
according to Listing~\ref{lst:unr-1}.
The constant normal forms are $\m f$ and $\m b$,
so we add $([\m f],[\m f],\m f)$ and $([\m b],[\m b],\m b)$
to the \textit{worklist}.
Then we enter the main loop.
We may ignore duplicate entries on the \textit{worklist},
because they are skipped on line~\ref{d11}.
Taking this into account, the main loop is executed 3 times.
\begin{enumerate}[align=left]
\item
$([\m f],[\m f],\m f)$ is taken from the \textit{worklist},
and we assign $w([\m f]) = \m f$ and $n([\m f]) = \m f$.
The conditions on lines~\ref{d14} and~\ref{d15} are never satisfied.
Because $[\m f] \to [\m f] \in \RCU$,
$([\m f],[\m f],\m f)$ is added to the \textit{worklist} again
by the final loop (lines~\ref{d21} and~\ref{d23}).
\item[(1$'$)]
$([\m f],[\m f],\m f)$ (which is a duplicate)
is taken from the \textit{worklist},
but now $w([\m f]) = \m f$ is defined and we reach line~\ref{d11}.
\item
$([\m b],[\m b],\m b)$ is taken from the \textit{worklist},
and we assign $w([\m b]) = \m b$ and $n([\m b]) = [\m b]$.
Line~\ref{d17} is not reached.
Line~\ref{d22} is reached for $[\m b] \to [\m b] \in \RCU$ and
$[\mfa] \to [\m b] \in \RCU$ and we add
$([\m b],[\m b],\m b)$ (a duplicate) and $([\mfa],[\m b],\m b)$ to the \textit{worklist}.
\item
$([\mfa],[\m b],\m b)$ is taken from the \textit{worklist},
and we assign $w([\mfa]) = \m b$ and $n([\mfa]) = [\m b]$.
$([\mfa],[\m b],\m b)$ (a duplicate)
is added to the \textit{worklist} on line~\ref{d22},
because $[\mfa] \to [\mfa] \in \RCU$.
\end{enumerate}
The loop terminates without reaching line~\ref{d9},
so the first \UNR-condition is satisfied for $\UU$.
We have derived the following exhaustive list of
instances of $s \in W(p,q)$ derivable by the rules in Figure~\ref{fig:unr-1},
corresponding to the normal forms $\m f$ and $\m b$.
\begin{xalignat*}{3}
\m f &\in W([\m f],[\m f]) &
\m b &\in W([\m b],[\m b]) &
\m b &\in W([\mfa],[\m b])
\end{xalignat*}

For $\VV$,
there is only one constant normal form, namely $\m f$.
Hence, initially, we add $([\m f]\EQU, [\m f], \m f)$ to the \textit{worklist}.
Then we enter the main loop (lines~\ref{d5}--\ref{d24}).
On line~\ref{d13}, $w([\m f])$ is set to $\m f$ and $n([\m f])$ is assigned $[\m f]$.
The only way that the check on line~\ref{d12} could succeed
would be having a rule with left-hand side $[\m f] \circ [\m f]$
in $\EEV$, which is not the case.
On line~\ref{d22}, we add $([\m f]\EQU, [\m f], \m f)$
to the \textit{worklist} again,
but in the next loop iteration, line~\ref{d11} is reached.
The loop terminates, having recorded the normal form $\m f$
with $\m f \in W([\m f],[\m f])$.
\end{example}

\begin{example}
\label{ex:exp-unr1}
We exhibit a class of TRSs with exponential witness size.
To this end, fix $k > 0$ and consider the rules
\begin{xalignat*}{2}
\m a_k &\to \m b &
\m a_i &\to \m a_{i-1} \circ \m a_{i-1}
\end{xalignat*}
where $0 < i \leq k$.
The check of the first \UNR-condition
will find the two normal forms $\m b$ and $t_k$ of $\m a_k$,
where $t_0 = \m a_0$ and $t_{i+1} = t_i \circ t_i$ for $0 < i \leq k$.
The term $t_k$ has size $2^{k+1}-1$,
but only $k-1$ distinct subterms.
\end{example}
\begin{remark}
The check of the first \UNR-condition (Listing~\ref{lst:unr-1})
is similar to the check of \UNC (Listing~\ref{lst:unc}),
with a few crucial differences:
\begin{itemize}
\item
First, we use the automaton $\x A = (\FF^{[]}, \FF^{[]}, \EER \cup \RCR^-)$
instead of $\CCR$.
\item
Because $\x A$ is not deterministic
(the rules of $\RCR^-$ are $\epsilon$-transitions),
different runs may result in the same term.
Hence the check on line~\ref{d8} is needed,
and witnesses need to be stored,
using maximal sharing for efficient equality tests.
\item
Furthermore, in addition to lines \ref{c11}--\ref{c17} in Listing~\ref{lst:unc},
which correspond to lines \ref{d14}--\ref{d20} in Listing~\ref{lst:unr-1},
we need a similar loop processing the $\epsilon$-transitions from $\RCR^-$,
cf.\ lines~\ref{d21}--\ref{d23} in Listing~\ref{lst:unr-1}.
The latter change increases the complexity from $\OO(\Vert\RR\Vert\log\Vert\RR\Vert)$ to $\OO(\Vert\RR\Vert^2)$.
\end{itemize}
\end{remark}
\begin{figure}
\begin{gather*}
\inferr[base$'$]{s \in W'(p,q)}{s \in W(p',q) & p \uparrow p'} \qquad
\inferr[step$'_\RCR$]{s \in W'(p',q)}{s \in W'(p,q) & p \to p' \in \RCR}
\\
\inferr[step$'_\EER$]{s_1 \circ s_2 \in W'(p,q)}{
s_1 \in W'(p_1,q_1) &
s_2 \in W'(p_2,q_2) &
p_1 \circ p_2 \to p \in \EER &
q_1 \circ q_2 \to q \in \NNR}
\end{gather*}
\caption{Inference rules for $s \in W'(p,q)$}
\label{fig:unr-2}
\end{figure}
\begin{listing}
\begin{algorithmic}[1]
\STATE check first \UNR-condition (obtaining $w(\cdot)$ and $n(\cdot)$), and compute $\uparrow$ \label{e1}
\STATE let $w'(p,q)$ be undefined for all $p \in \FF^{[]}$, $q \in \FF^{[]} \cup \{[\star]\}$ (to be updated below)
\FORALL{$p,q$ with $p \uparrow q$ and $w(q)$ defined} \label{e2}
 \STATE push $(p,n(q),w(q))$ to \textit{worklist} \label{e3}
\COMMENT{(\irule{base$'$})}
\ENDFOR \label{e4}
\WHILE{\textit{worklist} not empty} \label{e5}
 \STATE $(p,q,s) \gets$ pop \textit{worklist} \label{e6}
\COMMENT{$s \in W'(p,q)$}
 \IF[$t \in W'(p,q)$]{$w'(p,q) = t$ is defined} \label{e7}
  \IF{$s = t$ or $t = \infty$} \label{e8}
   \STATE \textbf{continue} at~\ref{e5}\label{e9}
  \ENDIF \label{e10}
  \STATE $w'(p,q) \gets \infty$, $s \gets \infty$ \label{e11}
\COMMENT{$|W'(p,q)| \geq 2$}
 \ELSE \label{e12}
  \STATE $w'(p,q) \gets s$ \label{e13}
 \ENDIF \label{e14}
 \IF{$w(p) = t$ is defined and $t \neq s$} \label{e15}
  \RETURN $\UNR(\RR)$ is false \label{e16}
\COMMENT{second \UNR-condition violated by $t$ and $s$}
 \ENDIF \label{e17}
 \FORALL{$p_1 \circ p_2 \to p_r \in \EER$
   and states $q_1,q_2$ of $\NNR$ with $(p,q) \in \{(p_1,q_1),(p_2,q_2)\}$} \label{e18}
  \IF{$w'(p_1,q_1) = s_1'$ and $w'(p_2,q_2) = s_2'$ are defined} \label{e19}
   \IF{there is a transition $q_1 \circ q_2 \to q_r \in \NNR$} \label{e20}
    \STATE push $(p_r,q_r,s_1' \circ s_2')$ to \textit{worklist} \label{e21}
\COMMENT{(\irule{step$'_\EER$})}
   \ENDIF \label{e22}
  \ENDIF \label{e23}
 \ENDFOR \label{e24}
 \FORALL{rules $p \to p' \in \RCR$} \label{e25}
  \STATE push $(p',q,s)$ to \textit{worklist} \label{e26}
\COMMENT{(\irule{step$'_\RCR$})}
 \ENDFOR \label{e27}
\ENDWHILE \label{e28}
\RETURN{$\UNR(\RR)$ is true} \label{e29}
\end{algorithmic}
\caption{Checking the second \UNR-condition.}
\label{lst:unr-2}
\end{listing}
For the second \UNR-condition,
let $W'(p,q)$ be the set of $\RR^\circ$-normal forms $t \in \TT(\FF^\circ)$
that are accepted by $\NNR$ in state $q$
and satisfy the right part of \eqref{peak-4}, i.e.,
\begin{equation}
\label{peak-6}
p \xl[\EER \cup \RCR]{*} C[p_1,\dots,p_n] \uparrow^\mnshortparallel C[q_1,\dots,q_n] \xl[\EER \cup \RCR^-]{*} t \xr[\NNR]{*} q
\end{equation}
\begin{lemma}
\label{lem:w'}
The predicate $s \in W'(p,q)$
is characterized by the inference rules in Figure~\ref{fig:unr-2}.
\end{lemma}
\proof
The inference rules follow by an inductive analysis on the left-most step of the
$p \xl[\EER \cup \RCR]{*} C[p_1,\dots,p_n]$ subreduction
of~\eqref{peak-6}:
\begin{itemize}[align=left]
\item[(\irule{base$'$})]
If the sequence is empty,
we have $p = p_1$, $p_1 \uparrow q_1$, and $t \in W(q_1,q)$.
Conversely,
we have $t \in W'(p,q)$ whenever $t \in W(p',q)$ and $p \uparrow p'$.
\item[(\irule{step$'_\RCR$})]
If the leftmost step is an $\RCR$ step, we have
$p \xl[\RCR]{} p' \xl[\EER \cup \RCR]{*} C[q_1,\dots,q_n]$,
and $t \in W'(p,q)$ for some $q$;
in that case, $t \in W'(p',q)$ follows.
\item[(\irule{step$'_\EER$})]
If the leftmost step is an $\EER$ step,
then either $p \xl[\EER]{} t \in \FF$,
but that case is already covered by (\irule{base$'$}),
or $p \xl[\EER]{} p_1 \circ p_2$, $t = t_1 \circ t_2$,
and there are states $q_1$, $q_2$
with $t_1 \in W'(p_1,q_1)$, $t_2 \in W'(p_2,q_2)$,
and a transition $q_1 \circ q_2 \to q_r \in \NNR$.
\qed
\end{itemize}
Because $W'(p,q)$ may be an infinite set,
we instead compute the partial function $w'(p,q)$ that returns
$s$ if $W'(p,q) = \{s\}$ is a singleton set,
or the special value $\infty$ if $W'(p,q)$ has at least two elements,
where $\infty$ is distinct from any term and satisfies
$\infty \circ t = t \circ \infty = \infty \circ \infty = \infty$
for all terms $t$.
The system has the \UNR property if $w'(p,q) = w(p)$ whenever $w'(p,q)$ and $w(p)$ are both defined.

The procedure is given in Listing~\ref{lst:unr-2}.
It maintains a \textit{worklist} of tuples $(p,q,s)$ where
$s$ is either a term with $s \in W'(p,q)$,
or has the special value $s = \infty$.
Note that each value $w'(p,q)$ may be updated up to two times:
it starts out as undefined, may be updated to an element of $W'(p,q)$,
and later to $\infty$ if $W'(p,q)$ has at least two elements.
These updates are performed by lines~\ref{e7}--\ref{e14}.

\begin{example}[continued from Examples~\ref{running-mc} and~\ref{running-unr1}]
\label{running-unr2}
Let us check the second \UNR-condition for
$\UU^\circ = \{ \mfa \to \m a, \mfa \to \m b, \m a \to \m a \}$
according to Listing~\ref{lst:unr-2}.
On line~\ref{e3}, we put the following items on the \textit{worklist}:
\begin{itemize}
\item
$([\m f],[\m f],\m f)$,
because $w([\m f]) = \m f$ and $[\m f]\uparrow[\m f]$;
\item
$([\m b],[\m b],\m b)$,
because $w([\m b]) = \m b$ and $[\m b]\uparrow[\m b]$,
and $w([\mfa]) = \m b$ and $[\m b]\uparrow[\mfa]$;
\item
$([\m a],[\m b],\m b)$,
because $w([\m b]) = \m b$ and $[\m a]\uparrow[\m b]$,
and $w([\mfa]) = \m b$ and $[\m a]\uparrow[\mfa]$; and
\item
$([\mfa],[\m b],\m b)$,
because $w([\m b]) = \m b$ and $[\mfa]\uparrow[\m b]$,
and $w([\mfa]) = \m b$ and $[\mfa]\uparrow[\mfa]$
\end{itemize}
Ignoring duplicates on the \textit{worklist}
(which are skipped on line~\ref{e9}),
the main loop (lines \ref{e5}--\ref{e28}) is executed 3 times.
\begin{enumerate}
\item
$([\m f],[\m f],\m f)$ is taken from the \textit{worklist},
and we let $w'([\m f],[\m f]) = \m f$.
$([\m f],[\m f],\m f)$ is added to the \textit{worklist} on line~\ref{e28}.
\item
$([\m a],[\m b],\m b)$ is taken from the \textit{worklist},
and we let $w'([\m a],[\m b]) = \m b$.
Now line~\ref{e21} is reached for
$[\m f] \circ [\m a] \to [\mfa] \in \EEU$,
$w'([\m f],[\m f]) = \m f$, $w'([\m a],[\m b]) = \m b$,
and $[\m f] \circ [\m b] \to [\star] \in \NNU$.
Hence $([\mfa],[\star],\m f \circ \m b)$ is added to the \textit{worklist}.
Furthermore, $([\m a],[\m b],\m b)$ is added to the \textit{worklist}
on line~\ref{e28} because $[\m a] \to [\m a] \in \RCU$.
\item
$([\mfa],[\star],\m f \circ \m b)$ is taken from the \textit{worklist}.
We set $w'([\mfa],[\star]) = \m f \circ \m b$.
The check on line~\ref{e15} succeeds
($w([\mfa]) = \m b \neq \m f \circ \m b$),
so the second \UNR-condition is not satisfied.
\end{enumerate}
We conclude that $\UNR(\UU)$ does not hold,
a fact witnessed by
$\m f \circ \m b \xl{} \m f \circ (\m f \circ \m a) \xr{} \m f \circ \m a \xr{} \m b$.

For $\VV$, $[\m f]\uparrow[\m f]$ is the only meetable constant relation
with defined witness $w(q)$,
so $([\m f],[\m f],\m f)$ is the only item ever pushed to the \textit{worklist}
(twice, because $[\m f] \to [\m f] \in \RCV$),
corresponding to $\m f \in W'([\m f],[\m f])$.
The second \UNR-condition holds for $\VV$,
implying $\UNR(\VV)$.
\end{example}

In order to achieve the desired complexity,
care must be taken with the enumeration on line~\ref{e18}:
instead of iterating over all elements of $\EER$,
an index mapping $p$ to the rules $p_1 \circ p_2 \to p_r \in \EER$
with $p \in \{p_1,p_2\}$ should be used.

The following two lemmas establish key invariants for showing
that the procedures in Listings~\ref{lst:unr-1} and~\ref{lst:unr-2}
faithfully implement the inference rules in Figures~\ref{fig:unr-1} and~\ref{fig:unr-2},
respectively.

\begin{lemma}
\label{lem:unr-inv-1}
Whenever line~\ref{d5} is reached in Listing~\ref{lst:unr-1}, we have
\begin{enumerate}
\item
If $(p,q,s) \in \textit{worklist}$ or $s = w(p)$ is defined and $q = n(p)$,
then $s \in W(p,q)$ holds.
\item
Assume that $\hat s \in W(\hat p,\hat q)$ can be inferred using
an inference rule from Figure~\ref{fig:unr-1}
with premises $P_1,\dots,P_n$.
Then either $(\hat p,\hat q,\hat s) \in \textit{worklist}$,
or $\hat s = w(\hat p)$ and $\hat q = n(\hat p)$,
or there is a premise $P_i = s' \in W(p',q')$
such that $w(p')$ is undefined or $s' \neq w(p')$.
\end{enumerate}
\end{lemma}
\proof
For the first invariant,
first note that $w(p)$ and $n(p)$ are updated simultaneously on line~\ref{d13}
with values that are taken from the \textit{worklist} on line~\ref{d6},
so we may focus on the addition of items to the \textit{worklist},
which happens on lines~\ref{d3}, \ref{d17} and~\ref{d22}.
On line~\ref{d3}, $c$ is a normal form,
and $c \in W([c],[c])$ holds by rule (\irule{base}).
On line~\ref{d17}, we have $s_1 \in W(p_1,q_1)$ and $s_2 \in W(p_2,q_2)$,
and rules $p_1 \circ p_2 \to p_r \in \EER$ and $q_1 \circ q_2 \to q_r \in \NNR$,
allowing us to infer $s_1 \circ s_2 \in W(p_r,q_r)$ by (\irule{step$_\EER$}).
Similarly, on line~\ref{d22},
we have $s \in W(p,q)$ and $p' \to p \in \RCR$,
and $s \in W(p',q)$ follows by (\irule{step$_\RCR$}).

Consider the second invariant
immediately after the loop on lines~\ref{d1}--\ref{d4}.
If $\hat s \in W(\hat p,\hat q)$ can be derived by (\irule{base}),
then it is put on the \textit{worklist} by that loop.
All other inferences have a premise $s' \in W(p',q')$
for which $w(p')$ is undefined, since $w(\cdot)$ is nowhere defined.
So initially, the invariant holds.
Noting that once $w(\hat p)$ is set, it will never be changed,
the invariant can be invalidated in only two ways.
\begin{enumerate}
\item
$(\hat p,\hat q,\hat s) \in \textit{worklist}$ is the item taken from the \textit{worklist}
on line~\ref{d6}.
In this case,
either the algorithm aborts early on line~\ref{d9},
or we reach line~\ref{d11},
which ensures $w(\hat p) = \hat s$ and $n(\hat p) = \hat q$ since $\hat s$ determines $\hat q$,
or we reach line~\ref{d13},
which assigns $w(\hat p) = \hat s$ and $n(\hat p) = \hat q$.
So the invariant is maintained.
\item
There is a premise $s' \in W(p',q')$
and $w(p')$ is assigned $s'$ on line~\ref{d13};
in that case, $(p',q',s')$ must be
the most recent item taken from the \textit{worklist} on line~\ref{d6}.
This can only happen if $\hat s \in W(\hat p,\hat q)$
is derived by one of the rules
(\irule{step$_\EER$}) or (\irule{step$_\RCR$}) in the last step.

If the (\irule{step$_\EER$}) rule is used,
let us assume that $s' = s_1$, $p' = p_1$ and $q' = q_1$
(the case that $s' = s_2$, $p' = p_2$ and $q' = q_2$ is completely analogous).
So $s_1 \in W(p_1,q_1)$ holds.
If the other premise $s_2 \in W(p_2,q_2)$ does not satisfy $w(p_2) = s_2$,
then the invariant remains true.
If both $w(p_1) = s_1$ and $w(p_2) = s_2$,
then $n(p_1) = q_1$ and $n(p_2) = q_2$
follow (because $\NNR$ is deterministic);
since $\hat s \in W(\hat p,\hat q)$ is derivable by (\irule{step$_\EER$}),
there must also be rules $p_1 \circ p_2 \to \hat p \in \EER$
(where $(p',q')$ is one of $(p_1,q_1)$ or $(p_2,q_2)$)
and $q_1 \circ q_2 \to \hat q \in \NNR$.
Consequently, $(\hat p = p_r,\hat q = q_r,\hat s = s_1 \circ s_2)$ will be added to the \textit{worklist} on line~\ref{d17}.

If the (\irule{step$_\RCR$}) is used,
$\hat q = q'$ holds and
there must be a step $\hat p \to p' \in \RCR$;
hence $(\hat p,\hat q,\hat s)$ will be put on the \textit{worklist} on line~\ref{d22}.
\qed
\end{enumerate}

\begin{lemma}
\label{lem:unr-inv-2}
Whenever line~\ref{e5} is reached in Listing~\ref{lst:unr-2}, we have
\begin{enumerate}
\item
If $(p,q,s) \in \textit{worklist}$ or $s=w'(p,q)$ is defined,
then $s \in W'(p,q)$ or $s = \infty$ and $|W'(p,q)| > 1$.
\item
Assume that $\hat s \in W'(\hat p,\hat q)$ can be inferred using
an inference rule from Figure~\ref{fig:unr-2}
with premises $P_1,\dots,P_n$.
Then $(\hat p,\hat q,\hat s) \in \textit{worklist}$,
$w'(\hat p,\hat q) = \hat s$, $w'(\hat p,\hat q) = \infty$,
or there is a premise $P_i = s' \in W'(p',q')$ such that $w'(p',q')$ is
not equal to $s'$ or $\infty$.
\end{enumerate}
\end{lemma}

\proof
Consider the first invariant.
First note that $w'(p,q)$ is only updated on lines~\ref{e11} and~\ref{e13}.
In this case, $(p,q,s) \in \textit{worklist}$ was true
at the beginning of the loop,
so $s \in W'(p,q)$ or $s = \infty$ and $|W'(p,q)| > 1$.
This justifies setting $w'(p,q) = s$ on line~\ref{e13}.
On line~\ref{e11}, we additionally have $t \in W'(p,q)$;
we conclude that $|W'(p,q)| > 1$
(justifying $w'(p,q) = \infty$)
because either $s = \infty$, or $s \neq t$ and $s,t \in W'(p,q)$.

Hence we may focus on the items put on the \textit{worklist}.
On line~\ref{e3}, since $w(q) \in W(q,n(q))$ and $p \uparrow q$,
$w(q) \in W'(p,n(q))$ follows.
On line~\ref{e21}, we have $p_1 \circ p_2 \to p_r \in \EER$
and $q_1 \circ q_2 \to q_r \in \EER$.
Moreover, we have either $s_1' \in W'(p_1,q_1)$ or $s_1' = \infty$ and $|W'(p_1,q_1)| > 1$; and
either $s_2' \in W'(p_2,q_2)$ or $s_2' = \infty$ and $|W'(p_2,q_2)| > 1$.
If neither $s_1' = \infty$ nor $s_2' = \infty$,
then we have $s_1' \circ s_2' \in W'(p,q)$ by (\irule{step$_\EER$}).
Otherwise, since we can derive $t_1 \circ t_2 \in W'(p,q)$ by (\irule{step$_\EER$})
for any $t_1 \in W'(p_1,q_1)$ and $t_2 \in W'(p_2,q_2)$,
$|W'(p,q)| > 1$ follows, and $s_1' \circ s_2' = \infty$.
So the invariant holds.
Finally, on line~\ref{e26}, we have $p \to p' \in \RCR$,
and
either $s \in W'(p,q)$ or $s = \infty$ and $|W'(p,q)| > 1$.
In the former case, $s \in W'(p',q)$ by (\irule{step$_\RCR$}),
while in the latter case, $t \in W'(p',q)$ for any $t \in W'(p,q)$;
either way, the invariant holds.

Next we consider the second invariant.
Immediately after the loop on lines~\ref{e2}--\ref{e4},
if $\hat s \in W'(\hat p,\hat q)$ follows by (\irule{base$'$}) then
$(\hat p,\hat q,\hat s)$ will be on the \textit{worklist}.
For all other inferences of some $\hat s \in W'(\hat p,\hat q)$,
there is a premise $s' \in W'(p',q')$ such that $w'(p',q')$ is undefined,
because initially, $w'$ is nowhere defined.
The invariant for $\hat s \in W'(\hat p,\hat q)$ may be invalidated in three ways.
\begin{enumerate}
\item
$(\hat p,\hat q,\hat s)$ is
the item taken from the \textit{worklist} on line~\ref{e6}.
In this case, lines~\ref{e7}--\ref{e14} ensure
that $w'(\hat p,\hat q) = \hat s$ or $w'(\hat p,\hat q) = \infty$
at the next loop iteration.
\item
$w'(\hat p,\hat q) = \hat s$ holds and the value of $w'(\hat p,\hat q)$ is updated;
this may only happen on line~\ref{e11},
and the invariant still holds with $w'(\hat p,\hat q) = \infty$.
\item
There is a premise $s' \in W'(p',q')$,
and $w'(p',q')$ is set to $s'$ or $\infty$ on line~\ref{e11} or~\ref{e13}.
This means that the item taken from the \textit{worklist} on line~\ref{e6}
satisfies $p = p'$, $q = q'$, and $s = s'$ or $s = \infty$.
Note that $s = w'(p,q)$ holds at line~\ref{e15}.

If (\irule{step$'_{\EER}$}) is used to infer $\hat s \in W'(\hat p,\hat q)$,
let us assume that $s' = s_1$, $p' = p_1$ and $q' = q_1$;
(the case that $s' = s_2$, $p' = p_2$ and $q' = q_2$ is analogous).
If $w'(p_2,q_2)$ is not equal to $s_2$ nor $\infty$,
then the invariant is maintained.
Note that we have $p_1 \circ p_2 \to \hat p \in \EER$ and
$q_1 \circ q_2 \to \hat q \in \NNR$,
with $(p',q') \in \{(p_1,q_1),(p_2,q_2)\}$.
Hence line~\ref{e21} is reached with $p_r = \hat p$, $q_r = \hat q$.
At that point, $s_1 = w'(p_1,q_1)$, $s_2 = w'(p_2,q_2)$.
So either $s_1 \circ s_2 = \hat s$, or $s_1 \circ s_2 = \infty$
so putting the item $(p_r,q_r,s_1 \circ s_2)$ on the \textit{worklist}
restores the invariant.

If (\irule{step$'_{\RCR}$}) is used to infer $\hat s \in W'(\hat p,\hat q)$,
then $q' = \hat q$ and there must be a rule $p' \to \hat p \in \RCR$.
Hence line~\ref{e26} will put $(\hat p,\hat q,s)$ with
$s = \hat s$ or $s = \infty$ on the \textit{worklist},
restoring the invariant.
\qed
\end{enumerate}

\begin{theorem}
The procedure in Figures~\ref{lst:unr-1} and~\ref{lst:unr-2} is correct
and takes $\OO(\Vert\RR\Vert^3 \log \Vert\RR\Vert)$ time.
\end{theorem}

\begin{proof}
Consider the check of the first \UNR-condition (Listing~\ref{lst:unr-1}).
Note that by the first invariant of Lemma~\ref{lem:unr-inv-1},
the check on line~\ref{d8} succeeds only
if the first \UNR-condition is violated,
since at that point $s \in W(p,q)$, $t \in W(p,n(q))$ and $s \neq t$,
so $s \xr[\EER \cup \RCR^-]{*} p \xl[\EER \cup \RCR^-]{*} t$ and
$s$ and $t$ are $\RR^\circ$-normal forms.
When the main loop exits, the \textit{worklist} is empty,
making the case that $(p,q,s) \in \textit{worklist}$
in the second invariant impossible.
Therefore,
we can show by induction on the derivation
that for any derivation of $s \in W(p,q)$
by the inference rules in Figure~\ref{fig:unr-1},
$w(p) = s$ holds,
using the second invariant.
Consequently, the resulting partial function $w$
witnesses the fact that the first \UNR-condition holds.
Therefore, the check of the first \UNR-condition is correct.

Now look at the check of the second \UNR-condition (Listing~\ref{lst:unr-2}).
Using the first invariant of Lemma~\ref{lem:unr-inv-2},
we see that the check on line~\ref{e15}
succeeds only if the second \UNR-condition is violated,
since at that point, either $s \in W'(p,q)$ and $s \neq w(p)$,
or $|W'(p,q)| > 1$,
ensuring that $W'(p,q)$ contains an element distinct from $w(p)$.
On the other hand,
if we reach line~\ref{e29},
the \textit{worklist} is empty,
and by induction on the derivation we can show that
for all derivations of $s \in W'(p,q)$ using the rules in Figure~\ref{fig:unr-2},
either $w'(p,q) = s$ or $w'(p,q) = \infty$,
using the second invariant.
Furthermore, the check on line~\ref{e15} has failed
for all defined values of $w'(p,q)$,
which means that whenever both $w'(p,q)$ and $w(p)$ are defined,
then they are equal; in particular, $w'(p,q) \neq \infty$.
Therefore, in~\eqref{peak-6}, if $w(p)$ is defined,
we must have $t = w'(p,q) = w(p) = s$,
and the second \UNR-condition follows, establishing \UNR by Lemma~\ref{lem-unr}.

Next we establish the complexity bound. Let $n = \Vert\RR\Vert$.
We claim that the check of the first \UNR-condition (Listing~\ref{lst:unr-1})
takes $\OO(n^3)$ time.
First note that the precomputation (line~\ref{d1}) can be performed cubic time.
Moreover, the bottom part of the main loop (lines~\ref{d13}--\ref{d23})
is executed at most $\OO(n)$ times, once for each possible value of $p$.
So even without indexing the rules of $\EER$,
the bottom part takes at most $\OO(n^2 \log n)$ time,
where the $\log n$ factor stems from the query of $\NNR$ and
the maintenance of maximal sharing when constructing $s_1 \circ s_2$.
Furthermore, only $\OO(n^2)$ items are ever added to the \textit{worklist},
so the top part of the loop (lines~\ref{d5}--\ref{d12}) also takes
$\OO(n^2)$ time. Overall, the check of the first \UNR-condition is
dominated by the cubic time precomputation.

For the complexity second \UNR-condition (Listing~\ref{lst:unr-2}),
note that computation of $\uparrow$ takes $\OO(n^3)$ time.
We focus on the main loop (lines~\ref{e4}--\ref{e28}).
Because $w'(p,q)$ is updated at most twice for each combination $(p,q)$,
lines~\ref{e15}--\ref{e27} are executed at most $\OO(n)$ times for
each possible value of $p$, for a total of $\OO(n^2)$ times.
By indexing the rules of $\EER$ we
can perform the enumeration on line~\ref{e19} in a total $\OO(n^3)$ time,
accounting for $\OO(n^2)$ selected rules of $\EER$
(each of which is used for at most two values of $p$),
and $\OO(n)$ possible values for $q_1$ or $q_2$,
depending on whether $(p,q) = (p_1,q_1)$ or $(p,q) = (p_2,q_2)$.
By the same analysis, lines~\ref{e20}--\ref{e24} are also executed $\OO(n^3)$ time in total,
for a total runtime of $\OO(n^3 \log n)$
(as for the first check, the $\log n$ factor
stems from the query of $\NNR$
and the maximal sharing of $s_1 \circ s_2$).
Lines~\ref{e25}--\ref{e27} are also executed $\OO(n^3)$ times.
Overall at most $\OO(n^3)$ items are added to \textit{worklist},
so lines~\ref{e6}--\ref{e17},
which are executed once per \textit{worklist} item,
take $\OO(n^3)$ time.
In summary, the complexity is $\OO(n^3 \log n)$ as claimed.
\end{proof}

%%%%%%%%%%%%%%%%%%%%%%%%%%%%%%%%%%%%%%%%%%%%%%%%%%%%%%%%%%%%%%%%%%%%%%%%%%%%%%
\section{Deciding \NFP}
\label{sec:nfp}

In this section we show how to decide
\NFP for a finite ground TRS $\RR$ in $\OO(\Vert\RR\Vert^3)$ time.
As preprocessing, we curry the TRS to bound its arity
(Section~\ref{sec:curry}),
compute the automaton $\NNR$ that accepts the $\RR^\circ$-normal forms
(Section~\ref{sec:nf}),
the congruence closure $\CCR$ for efficient convertibility checking
(Section~\ref{sec:cc}),
and the rewrite closure $(\RCR, \EER)$ that
allows testing reachability (Section~\ref{sec:rc}).

\begin{remark}
The decision procedure for \NFP is an extension of that for \UNC,
so reading Section~\ref{sec:unc} first is recommended.
\end{remark}

% % % % % % % % % % % % % % % % % % % % % % % % % % % % % % % % % % % % % % % 
\subsection{Conditions for \NFP}

For the analysis in this subsection,
we adapt the concept of top-stabilizable sides,
which was introduced for analyzing confluence of ground TRSs~\cite{CGN01,F12,GTV04}.
\begin{definition}
\label{def:ts}
A left-hand side $\ell$ of the transitions of $\CCR$
is \emph{top-stabilizable} if there is an $\EER/\RCR$-normal form $s$ with
$(s)\EQU \xr[\CCR]{*} \ell$.
\end{definition}

\begin{remark}
A term $s$ is top-stable (or root-stable)
if no reduction starting at $s$ has a root step.
In the particular case of $\RCR \cup \EER^\pm$,
all rules have a side in $\FF^{[]}$.
So $s$ is top-stable if there is no reduction from $s$ to an element of $\FF^{[]}$.
By Lemma~\ref{lem:rc},
any reduction $s \xr[\RCR \cup \EER^\pm]{*} q \in \FF^{[]}$
factors as $s \xr[\EER \cup \RCR]{*} q$,
so $s$ is top-stable if none of its $\EER/\RCR$-normal forms is in $\FF^{[]}$.
A term is top-stabilizable if it is convertible to a top-stable term.
In our setting,
convertibility is treated by the congruence closure $\CCR$,
which works over a different signature than the rewrite closure.
The connection is made by Lemma~\ref{lem:cc},
which explains why the operation $(\cdot)\EQU$
features in Definition~\ref{def:ts}.
\end{remark}

Assume that \NFP holds for $\RR^\circ$.
In particular, \UNC holds, and by the reasoning from Section~\ref{sec:unc},
there is a partial function $w(\cdot)$
mapping $q \in \FFR$
to the unique $\RR^\circ$-normal form that reaches $q$ by $\CC$ steps,
if such a normal form exists.
Because \NFP holds,
$s \xlr[\RR^\circ]{*} t$ with $s, t \in \TT(\FF^\circ)$ and
$t$ in $\RR^\circ$-normal form implies $s \xr[\RR^\circ]{*} t$.
The former is equivalent to $s \xr[\CCR]{*} u \xl[\CCR]{*} t$.
Let us assume that $u = q \in \FF\EQU$.
This implies $t = w(q)$,
so by Lemma~\ref{lem:cc} and Lemma~\ref{lem:rc} we obtain
a term $s'$ with
\[
s \xr[\CCR]{*} q \xl[\CCR]{*} t = w(q)
\qquad
\text{and}
\qquad
s \xr[\EER/\RCR]{*} s' \xl[\EER/\RCR^-]{*} t
\]
Assume that $s'$ is an $\EER/\RCR$-normal form of $s$.
Note that if $t = t_1 \circ t_2$ (i.e., $t$ is not a constant),
then there are states $q_1$, $q_2$ of $\CCR$
such that $t_1 = w(q_1)$, $t_2 = w(q_2)$ and
$t \xr[\CCR]{*} q_1 \circ q_2 \xr[\CCR]{} q$.
Further note that if $s' = s'_1 \circ s'_2$ (i.e., $s'$ is not a constant),
then we have states $p_1$, $p_2$ of $\CCR$ with
$(s')\EQU \xr[\CCR]{*} p_1 \circ p_2  \xr[\CCR]{} q$,
which means that $p_1 \circ p_2$ is top-stabilizable.
Conversely, $p_1 \circ p_2 \to q \in \CCR$ with
a top-stabilizable side $p_1 \circ p_2$ implies the existence of such
a term $s$ with non-constant $\EER/\RCR$-normal form $s'$
($s'$ exists by definition and we can let $s = s'\nf[\EER^-]$).
If $s' \in \FF$ then $s = s' = t$.
We consider four remaining cases with $s' \notin \FF$,
based on whether or not $s'$ and $t$ are constants.
\begin{enumerate}
\item
$s' \in \FF^{[]}$ and $t \in \FF$.
Then $s' \xr[\RCR]{} \cdot \xr[\EER^-]{} t$
(using that $\RCR$ is reflexive and transitive).
\item
$s' = s'_1 \circ s'_2$, and $t \in \FF$.
We have a contradiction to $t \xr[\EER/\RCR^-]{*} s'$.
\item
$s' \in \FF^{[]}$, and $t = t_1 \circ t_2$.
Then there must be a rule $p'_1 \circ p'_2 \to p' \in \EER$ such that
$p'_1 \circ p'_2 \xr[\EER]{} p' \xr[\RCR^-]{} s'$,
$q_1 = (p'_1)\EQU$ and $q_2 = (p'_2)\EQU$.
(The latter two conditions are necessary for
$t_1 \circ t_2 \xr[\EER/\RCR]{*} p'_1 \circ p'_2$ to hold,
cf.\ Lemma~\ref{lem:cc})
\item
$s' = s'_1 \circ s'_2$ and $t = t_1 \circ t_2$.
Then $s'_1 \xr[\EER/\RCR]{*} t_1$,
$s'_2 \xr[\EER/\RCR]{*} t_2$ imply $p_1 = q_1$ and $p_2 = q_2$
by Lemma~\ref{lem:cc}.
\end{enumerate}
From these four cases we obtain the following necessary conditions,
using the fact that $w(q) = t$.
\begin{definition}
The \NFP-conditions are
\begin{enumerate}
\item
if $s' \in \FF^{[]}$ and $t = w(q) \in \FF$ are constants,
then $s' \xr[\RCR]{} {\cdot} \xr[\EER^-]{} t$.
\item
if $p_1 \circ p_2 \to q \in \CCR$, and
$w(q) \in \FF$ is a constant,
then $p_1 \circ p_2$ is not top-stabilizable.
\item
if $s' \in \FF^{[]}$ is a constant,
$w(q_1) \circ w(q_2)$ is a $\RR^\circ$-normal form, and
$q_1 \circ q_2 \to q \in \CCR$,
then there is a rule $p'_1 \circ p'_2 \to p' \in \EER$ with
$p' \xr[\RCR^-]{} s'$ and $q_1 = (p'_1)\EQU$ and $q_2 = (p'_2)\EQU$.
\item
if $p_1 \circ p_2 \to q \in \CCR$
with top-stabilizable $p_1 \circ p_2$,
and $w(q_1) \circ w(q_2)$ is a $\RR^\circ$-normal form, then
$p_1 = q_1$ and $p_2 = q_2$.
\end{enumerate}
\end{definition}

\begin{lemma}
The \NFP-conditions are necessary and sufficient for \NFP to hold,
provided that $\RR^\circ$ is \UNC.
\end{lemma}

\begin{proof}
Necessity has already be established.
Assume that $\RR^\circ$ is \UNC and satisfies the \NFP-conditions.
Let $s \xlr[\RR^\circ]{*} t$ with $s,t \in \TT(\FF^\circ)$ and $t$ in normal form.
We have $s \xr[\CCR]{*} u \xl[\CCR]{*} t$.
Let $s'$ be an $\EER/\RCR$-normal form of $s$
(note that each application of an $\EER$-rule decreases the size of the term,
while $\RCR$-rules do not change the size of terms,
so $\EER/\RCR$ is terminating).
We have $s' \xlr[\EER\cup\RR^\flat]{*} t$, and
consequently $(s')\EQU \xr[\CCR]{*} u \xl[\CCR]{*} (t)\EQU = t$
by Lemma~\ref{lem:cc}.
In the remainder of the proof,
we show that $s' \xr[\EER^-/\RCR]{*} t$.
This will establish $s \xr[\RR^\circ]{*} t$
by Lemma~\ref{lem:rc} in conjunction with Proposition~\ref{prop:flat}:
\[
s = s\nf[\EER^-] \xr[\RR^\circ]{*} s'\nf[\EER^-] \xr[\RR^\circ]{*} t\nf[\EER^-] = t
\]
We proceed by induction on $t$.
If $u \notin \FF^{[]}$,
then either $(s)\EQU = u = t$ are constants,
or
\[
(s)\EQU = (s_1)\EQU \circ (s_2)\EQU \xr[\CCR]{*} u_1 \circ u_2 \xl[\CCR]{*} t_1 \circ t_2 = t
\]
for some terms $s_1$, $s_2$, $t_1$, $t_2$,
and we conclude by the induction hypotheses for $t_1$ and $t_2$.
So assume that $u = q \in \FF^{[]}$.
Then we have $(s')\EQU \xr[\CCR]{*} q \xl[\CCR]{*} t = w(q)$.
We consider three cases.
\begin{enumerate}
\item
If $(s')\EQU$ and $t$ are constants, then so is $s'$.
Because $s'$ is an $\EER$-normal-form,
we must have $s' \in \FF^{[]}$.
By the first \NFP-condition,
$s \xr[\EER/\RCR]{*} s' \xr[\RCR]{} {\cdot} \xr[\EER^-]{} t$,
which implies $s \xr[\RR^\circ]{*} t$
by Lemma~\ref{lem:rc} and Proposition~\ref{prop:flat}.
\item
If $(s')\EQU$ is not constant,
then $s' = s_1' \circ s_2'$ for some terms $s_1'$, $s_2'$,
and there are $p_1,p_2 \in \FFR$ with
$(s_1')\EQU \circ (s_2')\EQU \xr[\CCR]{*} p_1 \circ p_2 \xr[\CCR]{} q$.
This means that $p_1 \circ p_2$ is a top-stabilizable side,
so $t$ cannot be a constant by the second \NFP-condition.
Hence there are terms $t_1$, $t_2$ with $t = t_1 \circ t_2$,
and $q_1, q_2 \in \FFR$ such that
$t_1 \circ t_2 \xr[\CCR]{*} q_1 \circ q_2 \xr[\CCR]{} q$.
By the fourth \NFP-condition, we have $p_1 = q_1$ and $p_2 = q_2$,
and consequently, $(s'_i) \xr[\CCR]{*} q_i \xl[\CCR]{*} t_i$
for $i \in \{1,2\}$. We conclude by the induction hypothesis.
\item
If $(s')\EQU$ is a constant, but $t$ is not,
then there are terms $t_1$, $t_2$ with $t = t_1 \circ t_2$,
and $q_1, q_2 \in \FFR$ such that
$t_1 \circ t_2 \xr[\CCR]{*} q_1 \circ q_2 \xr[\CCR]{} q$.
By the third \NFP-condition, we obtain $p'_1$, $p'_2$ such that
$s' \xr[\EER/\RCR]{} p'_1 \circ p'_2$ and
$(p'_i) = q_i \xl[\CCR]{*} t_i$ for $i \in \{1,2\}$,
and we conclude by the induction hypothesis.
\qedhere
\end{enumerate}
\end{proof}

% % % % % % % % % % % % % % % % % % % % % % % % % % % % % % % % % % % % % % %
\subsection{Computing Top-Stabilizable Sides}
\label{sec:ts}

First note that for rules $c \to [c]\EQU \in \CCR$ with $c \in \FF$,
$c$ is never top-stabilizable, since $(s)\EQU \xr[\CCR]{*} c$ implies $s = c$,
and $c \xr[\EER]{} [c]$.
So any top-stabilizable side must have shape $p \circ q$.
In order to compute the top-stabilizable sides,
let us first consider $\EER/\RCR$ reducible terms
of shape $p \circ q$ with $p,q \in \FF^{[]}$.
Because $\RCR$ is reflexive and transitive,
this means that there are rules $p \to p' \in \RCR$ and $q \to q' \in \RCR$
such that $p' \circ q'$ is a left-hand side of $\EER$.
We can compute
\[
\NF^\circ = \{p \circ q \mid p,q \in \text{$\FF^{[]}$ and $p \circ q$ is in $\EER/\RCR$-normal form}\}
\]
as the complement of the $\EER/\RCR$-reducible terms;
the latter can be computed by enumerating the $\OO(\Vert\RR\Vert)$
left-hand sides $p \circ q$ of $\EER$,
and the $\OO(\Vert\RR\Vert^2)$ possible pairs $(p',q')$,
taking $\OO(\Vert\RR\Vert^3)$ time in total.
The size of $\NF^\circ$ is $\OO(\Vert\RR\Vert^2)$.

Let $\TS(p \circ q)$ denote the fact that $p \circ q$ is a top-stabilizable side.
For convenience, we extend the notion to the right-hand sides of $\CCR$:
$\TS(q)$ for $p_1 \circ p_2 \to q \in \CCR$ with $\TS(p_1 \circ p_2)$.
In this case, we call $q$ a top-stabilizable constant.
The top-stabilizable constants and sides can be found using an
incremental computation.
Every $\EER/\RCR$-normal form $p \circ q \in \NF^\circ$
for which $(p)\EQU \circ (q)\EQU$ is a left-hand side of $\CCR$
induces a top-stabilizable side $(p)\EQU \circ (q)\EQU$.
If $p_1 \circ p_2$ is top-stabilizable and $p_1 \circ p_2 \xr[\CCR]{} q$, then
$q$ is a top-stabilizable constant.
For any top-stabilizable constant $p$, $p \circ q$, $q \circ p$ for
constant $q \in \FF\EQU$ are also top-stabilizable.
Consequently, we obtain the following inference rules,
where $i \in \{1,2\}$.
\[
\inferr[nf]{\TS((p_1)\EQU \circ (p_2)\EQU)}{p_1 \circ p_2 \in \NF^\circ}
\qquad
\inferr[ts$_0$]{\TS(q)}{p_1 \circ p_2 \to q \in \CCR & \TS(p_1 \circ p_2)}
\qquad
\inferr[ts$_i$]{\TS(q_1 \circ q_2)}{\TS(q_i)}
\]
Restricting to left-hand sides of $\CCR$,
there are $\OO(\Vert \RR \Vert)$
instances of (\irule{nf}), (\irule{ts$_0$}),
(\irule{ts$_1$}) and (\irule{ts$_2$}).
Using Horn inference, they allow computing the top-stabilizable sides
in $\OO(\Vert \RR \Vert)$ time.
The run time is dominated by the computation of $\NF^\circ$,
which takes $\OO(\Vert \RR \Vert^3)$ time.

\begin{example}[continued from Examples~\ref{running-cc} and~\ref{running-rc}]
\label{running-ts}
For $\UU$ the $\EEU/\RCU$-reducible terms are
$\NNF_\UU = \{ [\m f] \circ [\m a], [\m f] \circ [\mfa] \}$.
Note that normal forms also include terms like
$[\m{f}] \circ [\m{f}]$ or $[\mfa] \circ [\m{a}]$
that have no correspondence in the original TRS.
For $\VV$ we obtain the following $\EEV/\RCV$-reducible terms:
$\NNF_\VV = \{[\m{f}] \circ [\m{b}], [\m{f}] \circ [\mfa],
 [\m{f}] \circ [\mfb],[\m{f}] \circ [\m{a}], [\m{f}] \circ [\mffb],
 [\m{f}] \circ [\mfffb] \}$.

In the $\UU$ case,
we have $\TS([\m f]_\UU \circ [\m a]_\UU)$
(e.g., because $[\m f] \circ [\m b] \in \NF_\UU$)
and $\TS([\m a]_\UU)$,
whereas for $\VV$ no top-stabilizable sides or constants exist.
\end{example}

% % % % % % % % % % % % % % % % % % % % % % % % % % % % % % % % % % % % % % % 
\subsection{Checking \NFP}

We base the procedure on
the decision procedure for \UNC (Listing~\ref{lst:unc}).
A closer look at the \NFP-condition reveals
that they are fairly easy to check,
provided one starts with an enumeration of all pairs $(w(q),q)$
where $w(q)$ is defined.
The \UNC decision procedure works by doing exactly that:
each item pushed to its \textit{worklist} corresponds to
a pair $(w(q),q)$,
and if \UNC holds, each such pair is enumerated exactly once.
Therefore we can modify the procedure to check \NFP instead of \UNC,
see Listing~\ref{lst:nfp}.

\begin{listing}
  % on a black white printout, the blue looks identical to black, so I've
  % changed the color to gray
\definecolor{newlines}{RGB}{120,120,120}
\newcommand{\NEW}[1]{\textcolor{newlines}{#1}}
\begin{algorithmic}[1]
\STATE compute $\CCR$ and a representation of $\NNR$ \label{f1}
\STATE let $\mathit{seen}(p)$ be undefined for all $p \in \FFR$ (to be updated below)
%NEW
\color{newlines}
\STATE compute $\RCR$ and top-stabilizable sides $\TS$ \label{f2}
\color{black}
%END NEW
\FORALL{constants $c \trianglelefteq \RR^\circ$ that are normal forms} \label{f3}
 \STATE push $([c]\EQU,[c], c)$ to \textit{worklist} \label{f4}
%NEW
\color{newlines}
 \IF {any constant $p \in \FF^{[]}$ with $(p)\EQU = [c]\EQU$
   does not satisfy $p \xr[\RCR]{} [c]$} \label{f5}
  \RETURN {$\NFP(\RR)$ is false} \COMMENT{first \NFP-condition violated} \label{f6}
 \ENDIF \label{f7}
 \IF {there is any $p_1 \circ p_2 \in \TS$ with
  $p_1 \circ p_2 \to [c]\EQU \in \CCR$} \label{f8}
  \RETURN {$\NFP(\RR)$ is false} \COMMENT{second \NFP-condition violated} \label{f9}
 \ENDIF \label{f10}
\color{black}
%END NEW
\ENDFOR \label{f11}
\WHILE{\textit{worklist} not empty} \label{f12}
 \STATE $(p,q,s) \gets$ pop \textit{worklist} \label{f13}
 \IF{$\mathit{seen}(p)$ is defined} \label{f14}
  \RETURN $\NEW{\NFP}(\RR)$ is false \COMMENT{\NEW{not \UNC}}\label{f15}
 \ENDIF \label{f16}
 \STATE $\mathit{seen}(p) \gets (q,s)$ \label{f17}
 \FORALL{transitions $p_1 \circ p_2 \to p_r \in \CCR$ with $p \in \{p_1,p_2\}$} \label{f18}
  \IF{$\mathit{seen}(p_1) = (q_1,s_1)$ and $\mathit{seen}(p_2) = (q_2,s_2)$ are defined} \label{f19}
   \IF{there is a transition $q_1 \circ q_2 \to q_r \in \NNR$} \label{f20}
    \STATE push $(p_r,q_r,s_1 \circ s_2)$ to \textit{worklist} \label{f21}
%NEW
\color{newlines}
    \STATE $G \gets \{ p' \mid p'_1 \circ p'_2 \to p' \in \EER,\,
      p_1 = (p'_1)\EQU,\, p_2 = (p'_2)\EQU \}$ \label{f22}
    \IF {there is $q \in \FF^{[]}$ with $(q)\EQU = p_r$ such that
      $q \not\xr[\RCR]{} q'$ for all $q' \in G$} \label{f23}
     \RETURN {$\NFP(\RR)$ is false} \COMMENT{third \NFP-condition violated} \label{f24}
    \ENDIF \label{f25}
    \IF {there is $p_1' \circ p_2' \in \TS$ with
      $p_1' \circ p_2' \to p_r \in \CCR$ and
      $p_1' \circ p_2' \neq p_1 \circ p_2$} \label{f26}
     \RETURN {$\NFP(\RR)$ is false} \COMMENT{fourth \NFP-condition violated} \label{f27}
    \ENDIF \label{f28}
\color{black}
%END NEW
   \ENDIF \label{f29}
  \ENDIF \label{f30}
 \ENDFOR \label{f31}
\ENDWHILE \label{f32}
\RETURN $\NEW{\NFP}(\RR)$ is true \label{f33}
\end{algorithmic}

\caption{Deciding $\NFP(\RR)$.
\NEW{Gray} parts differ from the \UNC procedure (Listing~\ref{lst:unc}).}
\label{lst:nfp}
\end{listing}

\begin{example}[continued from Examples~\ref{running-unc} and~\ref{running-ts}]
\label{running-nfp}
The underlying \UNC procedure executes in the same way as Example~\ref{running-unc},
but with additional checks whenever an item is added to the \textit{worklist}.
For $\UU$,
after $([\m a]_\UU,[\m b],\m b)$ is added to the \textit{worklist},
we find that $([\m b])_\UU = [\m a]_\UU$ but not $[\m b] \to [\m a] \in \RCU$,
so the first \NFP-condition is violated;
indeed we have $\m a \xlr{*} \m b$ but not $\m a \xr{*} \m b$.
We would also have a violation of the second \NFP-condition,
because $[\m f]_\UU \circ [\m a]_\UU \to [\m a] \in \CCU$ and
$\TS([\m f]_\UU \circ [\m a]_\UU)$.
A possible counterexample arising from this
is $\m f \circ \m b \xlr{*} \m b$, but not $\m f \circ \m b \xr{*} \m b$.

For $\VV$, only $([\m f]_\VV,[\m f],\m f)$ is added to the \textit{worklist},
and the corresponding checks of the first and second \NFP-conditions
succeed.
\end{example}

\begin{theorem}
\label{thm:nfp}
The procedure in Listing~\ref{lst:nfp}
decides $\NFP(\RR)$ in $\OO(\Vert\RR\Vert^3)$ time.
\end{theorem}

\begin{proof}
Listing~\ref{lst:nfp} is an extension of Listing~\ref{lst:unc}.
In particular note that if the procedure returns that $\NFP(\RR)$
is true, then line~\ref{f33} is reached,
so \UNC holds for the input TRS as well;
in other words, whenever $\UNC(\RR)$ is false,
the procedure returns false as well.
So let us assume that $\UNC(\RR)$ is true and we have
a corresponding partial function $w : \FF\EQU \to \TT(\FF^\circ)$
mapping convertibility classes to normal forms.

Compared to Listing~\ref{lst:unc},
Listing~\ref{lst:nfp} has additional checks on lines \ref{f5}--\ref{f10}
and \ref{f22}--\ref{f28}.
The enumeration on line~\ref{f4} covers
all constants $c$ and states $q \in \FF\EQU$ with $w(q) = c$,
noting that $c \xr[\EER]{} [c]\EQU$,
so $q = [c]\EQU$ is forced.
On lines~\ref{f5} and~\ref{f6},
we check the first \NFP-condition, where $s' = p$,
noting that the only $\EER^-$ step leading to $c$ is $[c] \xr[\EER^-]{} c$.
On lines~\ref{f8} and~\ref{f9},
the second \NFP-condition is checked.
By the proof of Theorem~\ref{thm:unc} (correctness of the \UNC procedure),
Line~\ref{f17} is executed exactly once for each state $q_r$
for which $w(q_r)$ is defined but not a constant: $w(q_r) = s_1 \circ s_2$.
Lines \ref{f22}--\ref{f25} check the third \NFP-condition
(the code uses $p_1 \circ p_2 \to p_r \in \CCR$ instead of
$q_1 \circ q_2 \to q_r \in \CCR$),
and lines~\ref{f26} and~\ref{f27} check the fourth \NFP-condition
(using $p'_1 \circ p'_2 \to p_r \in \CCR$ instead of
$p_1 \circ p_2 \to q \in \CCR$).

For analyzing the complexity let $n = \Vert\RR\Vert$.
Note that the additional precomputation takes cubic time,
and that the added checks in Listing~\ref{lst:nfp} are executed $\OO(n)$ times
(see the proof of Theorem~\ref{thm:unc}).
The most expensive addition is the check of the third \NFP-condition,
which may take $\OO(n^2)$ time each time it is executed,
for a total of $\OO(n^3)$.
Overall the computation time is $\OO(n^3)$ as claimed.
\end{proof}

\begin{remark}
As far as we know, this is the first polynomial time algorithm for
deciding \NFP for ground systems.
\end{remark}

\begin{remark}
Rather than making a copy of the \UNC procedure with the modifications
in Listing~\ref{lst:nfp},
one can parameterize the \UNC procedure with callbacks that are invoked
at lines~\ref{c3} and~\ref{c14},
to avoid duplication of code.
\end{remark}

%%%%%%%%%%%%%%%%%%%%%%%%%%%%%%%%%%%%%%%%%%%%%%%%%%%%%%%%%%%%%%%%%%%%%%%%%%%%%%
\section{Deciding Confluence}
\label{sec:cr}

We are given a finite ground TRS $\RR$ over a finite signature $\FF$.
As preprocessing, we curry and flatten the TRS in order to bound its arity
(Sections~\ref{sec:curry} and~\ref{sec:flat}),
obtaining $\RR^\circ$, $\EER$ and $\RR^\flat$.
We then compute the congruence closure $\CCR$ and rewrite closure $\RCR$,
enabling effective convertibility and reachability checking
(Sections~\ref{sec:cc} and~\ref{sec:rc}).
First we observe the following.

\begin{lemma}
\label{lem:rcr-cr}
The TRS $\RR^\circ$ is confluent
if and only if $\RCR \cup \EER^\pm$ is confluent.
\end{lemma}
\begin{proof}
Assume that $\RR^\circ$ is confluent.
By Lemma~\ref{lem:rc},
$\xr[\RR^\flat \cup \EER^\pm]{*}$ and $\xr[\RCR \cup \EER^\pm]{*}$ coincide,
so we show confluence (in fact, the Church-Rosser property)
of $\RR^\flat \cup \EER^\pm$.
Assume that $s \xlr[\RR^\flat \cup \EER^\pm]{*} t$.
This implies $s \nf[\EER^-] \xlr[\RR^\circ]{*} t \nf[\EER^-]$
by Proposition~\ref{prop:flat}.
Confluence of $\RR^\circ$ implies that there is a term $u$ with
$s \nf[\EER^-] \xr[\RR^\circ]{*} u \xl[\RR^\circ]{*} t \nf[\EER^-]$.
Using Proposition~\ref{prop:flat} again
we obtain a joining sequence for $s$ and $t$
using rules from $\RR^\flat \cup \EER^\pm$:
\[
s \xr[\EER^-]{*} s \nf[\EER^-]
\xr[\RR^\flat \cup \EER^\pm]{*} u
\xl[\RR^\flat \cup \EER^\pm]{*} t \nf[\EER^-]
\xl[\EER^-]{*} t
\]
Next assume that $\RCR \cup \EER^\pm$ is confluent.
By Lemma~\ref{lem:rc}, this implies confluence of $\RR^\flat \cup \EER^\pm$.
Assume that $s \xlr[\RR^\circ]{*} t$,
where we may assume that $s,t \in \TT(\FF^\circ)$,
because confluence is preserved by signature extension.
Then $s \xlr[\RR^\flat \cup \EER^\pm]{*} t$ by Proposition~\ref{prop:flat}.
Consequently, $s \downarrow_{\RR^\flat \cup \EER^\pm} t$ follows
by confluence of $\RR^\flat \cup \EER^\pm$.
Using Proposition~\ref{prop:flat} again,
\[
s = s \nf[\EER^-] \xr[\RR^\circ]{*} {\cdot} \xl[\RR^\circ]{*} t \nf[\EER^-] = t
\tag*{\qedhere}
\]
\end{proof}

% % % % % % % % % % % % % % % % % % % % % % % % % % % % % % % % % % % % % % %
\subsection{Conditions for Confluence}
In this subsection, we derive necessary conditions for
confluence of $\RR^\circ$ (and hence $\RR$),
and then show that they are sufficient as well.

We follow the approach in  \cite{GTV04} and \cite{Tiw02}, which
is based on the analysis of two convertible terms $s$, $t$ and their
normal forms with respect to a system of so-called \emph{forward rules}
of the rewrite closure, in our case using the system $\EER / \RCR$.
Let us assume that $\RR^\circ$ is confluent.
By Lemma~\ref{lem:rcr-cr}, $\RCR \cup \EE^\pm$ is confluent as well.
Let $s$ and $t$ be $\EER \cup \RCR$-convertible terms and
let $s'$ and $t'$ be $\EER/\RCR$-normal forms of $s$ and $t$.
Clearly, $s' \xlr[\RCR \cup \EER^\pm]{*} t'$ follows.
Equivalently, there is a term $u$ such that
\begin{equation}
\label{eq1}
(s')\EQU \xr[\CCR]{*} u \xl[\CCR]{*} (t')\EQU
\end{equation}
Let us assume that $u \in \FF\EQU$.
To capture the conditions on $s'$ and $t'$,
we make use of the concept of top-stabilizable sides (Definition~\ref{def:ts}).
If the sequence $(s')\EQU \xr[\CCR]{*} u$ is empty,
then $s' \in \FF^{[]}$;
otherwise, there are $s_1,s_2 \in \FF\EQU$ with $(s')\EQU \xr[\CCR]{*} s_1 \circ s_2 \xr[\CCR]{} u$,
which means that $s_1 \circ s_2$ is a top-stabilizable side.
Conversely,
for any $s_1 \circ s_2 \to u \in \CCR$
where $s_1 \circ s_2$ is top-stabilizable,
we obtain a $\EE/\RCR$-normal form $s'$ with $(s')\EQU \xr[\CCR]{*} u$
for which $(s')\EQU \xr[\CCR]{*} u$ holds.
An analogous analysis applies to $t'$.
Because $\RCR \cup \EER^\pm$ is confluent,
$s'$ and $t'$ are joinable,
which by Lemma~\ref{lem:rc} implies
\begin{equation}
\label{eq2}
s' \xr[\EER^- \cup \RCR]{*} {\cdot} \xl[\EER^- \cup \RCR]{*} t'
\end{equation}
Furthermore, by Lemma~\ref{lem:cc} and Proposition~\ref{prop:rcx},
\eqref{eq2} implies
\begin{equation}
\label{eq3}
(s')\EQU \xl[\CCR]{*} {\cdot} \xr[\CCR]{*} (t')\EQU
\end{equation}
We distinguish three cases.
\begin{enumerate}
\item
There are top-stabilizable sides
$s_1 \circ s_2$ and $t_1 \circ t_2$
with corresponding $\EER/\RCR$-normal forms $s'$ and $t'$
such that
\[
(s')\EQU \xr[\CCR]{*} s_1 \circ s_2 \xr[\CCR]{} u
\xl[\CCR]{} t_1 \circ t_2 \xl[\CCR]{*} (t')\EQU
\]
In this case, for $i \in \{ 1, 2 \}$,
the terms $s_i$ and $t_i$ are meetable by $\CCR$ steps by \eqref{eq3},
which implies $s_i = t_i$ noting that $\CCR$ is confluent
and that $s_i$ and $t_i$ are $\CCR$-normal forms.
\item
$t' \in \FF^{[]}$
and there is a top-stabilizable side $s_1 \circ s_2$
with associated $\EER/\RCR$-normal form $s'$
such that
\[
(s')\EQU \xr[\CCR]{*} s_1 \circ s_2 \xr[\CCR]{} u = (t')\EQU
\]
To satisfy \eqref{eq2} and \eqref{eq3}, there must be $t_1, t_2 \in \FF^{[]}$ such that
$t' \xr[\EER^-/\RCR]{} t_1 \circ t_2$, and
$(s_i)\EQU = (t_i)\EQU$ for $i \in \{ 1, 2 \}$.
\item
$s', t' \in \FF^{[]}$ with $(s')\EQU = (t')\EQU$.
Then $s' \downarrow_{\EER^- \cup \RCR} t'$.
\end{enumerate}
Hence we obtain the following necessary conditions for
confluence of $\RR^\circ$:

\begin{definition}
\label{def:concon}
The \emph{confluence conditions} for confluence of $\RCR \cup \EER^\pm$ are
as follows.
\begin{enumerate}
\item
If $s_1 \circ s_2$ and $t_1 \circ t_2$ are top-stabilizable sides such that 
$s_1 \circ s_2 \xr[\CCR]{} {\cdot} \xl[\CCR]{} t_1 \circ t_2$
then $s_i = t_i$ for $i \in \{1,2\}$.
\item
If $s_1 \circ s_2$ is a top-stabilizable side
and $s_1 \circ s_2 \to (t')\EQU \in \CCR$
for $t' \in \FF^{[]}$,
then there must be $t_1, t_2 \in \FF^{[]}$
such that $t' \xr[\EER^-/\RCR]{} t_1 \circ t_2$, and
$s_i = (t_i)\EQU$ for $i \in \{ 1, 2 \}$.
\item If $s', t' \in \FF^{[]}$ with $(s')\EQU = (t')\EQU$ then
$s' \downarrow_{\EER^- \cup \RCR} t'$.
\end{enumerate}
\end{definition}

\begin{lemma}
\label{lem:concon}
The confluence conditions are necessary and sufficient for confluence
of $\RR^\circ$.
\end{lemma}

\begin{proof}
Necessity has already been shown above.
For sufficiency, assume that the confluence conditions are satisfied.
We show confluence of $\RCR \cup \EER^\pm$,
from which confluence of $\RR^\circ$ follows by Lemma~\ref{lem:rcr-cr}.
Assume that there are terms $s,t$ that are convertible
($s \xlr[\RCR \cup \EER^\pm]{*} t$) but not joinable.
Then any corresponding $\EER/\RCR$-normal forms are not joinable either.
Let $s'$ and $t'$ be convertible $\EER/\RCR$-normal forms with no
common reduct such that $|s'|+|t'|$ is minimal.
By Lemma~\ref{lem:cc}, we have
\begin{equation}
\label{eqA}
(s')\EQU \xr[\CCR]{*} {\cdot} \xl[\CCR]{*} (t')\EQU
\end{equation}
First note that $s' \in \FF$ (or $t' \in \FF$) is impossible because
of the rules $c \to [c] \in \EE$ for $c \in \FF$.
We distinguish three cases.
\begin{enumerate}[align=left]
\item
If $s', t' \in \FF^{[]}$.
Then $(s')\EQU = (t')\EQU$ from \eqref{eqA}
and the fact that $(s')\EQU, (t')\EQU \in \FF\EQU$ are $\CCR$-normal forms.
Hence we obtain a joining sequence from the third confluence
condition, contradicting the non-joinability of $s'$ and $t'$.
\item
If $s' = s_1' \circ s_2'$ and $t' \in \FF^{[]}$, then
\eqref{eqA} becomes
$(s')\EQU \xr[\CCR]{*} s_1 \circ s_2 \xr[\CCR]{} (t')\EQU$,
noting that $(t')\EQU$ is an $\CCR$-normal form.
In particular, $s_1 \circ s_2$ is a top-stabilizable side.
By the second confluence condition we obtain
a term $t_1 \circ t_2$ such that $t' \xr[\EER^-/\RCR]{} t_1 \circ t_2$,
and $s_i = (t_i)\EQU$ for $i\in\{1,2\}$.
Because $(s_1')\EQU \xr[\CCR]{*} s_1 = (t_1)\EQU$,
$t_1$ and $s_1'$ are convertible by Lemma~\ref{lem:cc}.
Furthermore, since
$|t_1| + |s_1'| < |t'| + |s'|$, this implies that $t_1$ and $s_1'$
are joinable. Analogously, $t_2$ and $s_2'$ are also joinable, and
therefore $s'$ is joinable with $t_1 \circ t_2 \xl[\EER^-/\RCR]{} t'$,
contradicting our assumptions.
\item[(2$'$)]
The case that $s' \in \FF^{[]}$ and $t' = t_1' \circ t_2'$
is handled symmetrically.
\item
If $s' = s_1' \circ s_2'$ and
$t' = t_1' \circ t_2'$, then \eqref{eqA} becomes
$(s')\EQU \xr[\CCR]{*} r \xl[\CCR]{*} (t')\EQU$ for some $r$.
If $r \in \FF$ then $s' = r = t'$ is trivially joinable.
If $r = r_1 \circ r_2$ (i.e., $r$ is not a constant), then $s_1'$ and $t_1'$
are convertible via $r_1$ and
likewise $s_2'$ and $t_2'$ are convertible via $r_2$.
However, one of these pairs cannot be joinable, and we obtain a smaller
counterexample to confluence, a contradiction.
Therefore, we must have $r \in \FF\EQU$.
So \eqref{eqA} can be decomposed as
\[
(s')\EQU \xr[\CCR]{*} s_1 \circ s_2 \xr[\CCR]{} r
\xl[\CCR]{} t_1 \circ t_2 \xl[\CCR]{*} (t')\EQU
\]
In particular, $s_1 \circ s_2$ and $t_1 \circ t_2$ are top-stabilizable sides.
From the first confluence condition, we conclude that $s_1 = t_1$ and
therefore $s_1'$ and $t_1'$ are convertible.
By minimality of $|s'| + |t'|$, $s_1'$ and $t_1'$ must be joinable.
Analogously, $s_2'$ and $t_2'$ must also be joinable,
from which we conclude that $s' = s_1' \circ s_2'$ and $t' = t_1' \circ t_2'$
are joinable as well, a contradiction.
\qedhere
\end{enumerate}
\end{proof}

% % % % % % % % % % % % % % % % % % % % % % % % % % % % % % % % % % % % % % %
\subsection{Computation of Confluence Conditions}
First we compute all top-stabilizable sides of shape $u \circ v$.
We have already done this in the decision procedure for \NFP,
see Section~\ref{sec:ts}.
For the third confluence condition (cf.\ Definition~\ref{def:concon}),
we have to identify constants $p,q \in \FF^{[]}$ with
$p \downarrow_{\EER^- \cup \RCR} q$.

Note that $p \downarrow_{\EER^- \cup \RCR} q$ is equivalent to
$p \xl[\EER \cup \RCR^-]{*} {\cdot} \xr[\EER \cup \RCR^-]{*} q$.
This matches the definition of meetable constants (Definition~\ref{def:mc}),
with $\RCR$ replaced by $\RCR^-$.
Consequently,
we can use the inference rules from Section~\ref{sec:mc}
with $p \to q \in \RCR$ replaced by $q \to p \in \RCR$
for computing the joinable constants, see Figure~\ref{fig:jc}.
\begin{definition}
\label{def:jc}
We call $p,q \in \FF^{[]}$ \emph{joinable constants}
(written $p \downarrow q$)
if $p \downarrow_{\EER^- \cup \RCR} q$.
\end{definition}
\begin{figure}
\begin{gather*}
\inferr[refl]{p \downarrow p\mathstrut}{p \in \FF^{[]}}
\qquad
\inferr[cong]{p \downarrow q}{p_1 \circ p_2 \to p \in \EER & p_1 \downarrow q_1 & p_2 \downarrow q_2 & q_1 \circ q_2 \to q \in \EER}
\\
\inferr[step\/$_l$]{p \downarrow r}{p \to q \in \RCR & q \downarrow r & r \in \FF^{[]}}
\qquad
\inferr[step\/$_r$]{p \downarrow r}{p \in \FF^{[]} & p \downarrow q & r \to q \in \RCR}
\end{gather*}
\caption{Inference rules for joinable constants}
\label{fig:jc}
\end{figure}

With the precomputation done,
checking the confluence conditions is straightforward.
Note that the map from $(\cdot)\EQU : \FF^{[]} \to \FF\EQU$
that is needed for the second and third confluence conditions
is a byproduct of the congruence closure computation.

\begin{example}[continued from Example~\ref{running-rc}]
\label{running-jc}
The joinability relations for $\UU$ and $\VV$ are given below.
As in Example~\ref{running-rc}, the letters and superscripts indicate the
rule being used to derive the entry and computation stage.
\label{running5}
\[
{\downarrow}_\UU =
\begin{array}{r|ccccccc}
&\m{f}&\m{a}&\mfa&\m{b}\\
\hline
\m{f}& r^0 &     &     &     \\
\m{a}&     & r^0 &s^1_r&     \\
\mfa &     &s^1_l& r^0 &s^1_l\\
\m{b}&     &     &s^1_r& r^0 \\
\end{array}
\qquad
{\downarrow}_\VV =
\begin{array}{r|cccccccc}
&\m{f}&\m{a}&\mfa&\m{b}&\mfb&\mffb&\mfffb\\
\hline
\m{f} & r^0 &     &     &     &     &     &     \\
\m{a} &     & r^0 &t_l^1&t_l^1&t_l^1&t_l^1&t_l^1\\
\mfa  &     &t_r^1& r^0 &t_l^1& c^2 & c^2 & c^2 \\
\m{b} &     &t_r^1&t_r^1& r^0 &t_r^1& c^3 & c^3 \\
\mfb  &     &t_r^1& c^2 &t_l^1& r^0 & c^2 & c^4 \\
\mffb &     &t_r^1& c^2 & c^3 & c^2 & r^0 & c^3 \\
\mfffb&     &t_r^1& c^2 & c^3 & c^4 & c^3 & r^0 \\
\end{array}
\]
It is now easy to verify that $\UU$ violates the third confluence condition
(we have $[\m a]_\UU = [\m b]_\UU$ but not $[\m a]\downarrow_\UU [\m b]$),
and therefore non-confluent.
The other two confluence conditions are satisfied for $\UU$.
For the first confluence condition, note that $[\m f]_\UU \circ [\m a]_\UU$
is the only top-stabilizable side.
The second confluence condition follows from
the fact that $s_1 = [\m f]_\UU$, $s_2 = [\m a]_\UU$,
$t' \in \{ [\m a], [\mfa], [\m b] \}$, and
$[\m f] \circ [\m a] \xr[\EEU]{} [\mfa] \xr[\RCU]{} [\m a],[\m b]$.
On the other hand,
$\VV$ satisfies all confluence conditions and is, therefore, confluent.
\end{example}

Putting everything together, we obtain the following theorem.

\begin{theorem}
Confluence of a ground TRS $\RR$ can be decided in cubic time.
\end{theorem}

\begin{proof}
Let $n = \Vert \RR \Vert$. We follow the process outlined above.
First we curry $\RR$ in linear time, obtaining $\RR$ with
$\Vert \RR \Vert \in \OO(n)$.
Then we flatten $\RR$, obtaining
$(\RR^\flat,\EER)$ with $\Vert \EER \Vert \in \OO(n)$ and
$\Vert \RR^\flat \Vert \in \OO(n)$ in time $\OO(n \log n)$.
In the next step we compute the rewrite and congruence closures
$(\RCR,\EER)$ and $(\CCR,\EER)$ of $(\RR^\flat,\EER)$ in $\OO(n^3)$ time.
We then compute
the joinable constants relation $\downarrow$
(which are analogous to meetable constants),
and the top-stabilizable sides $\TS({-}\circ{-})$,
which as seen in Sections~\ref{sec:mc} and~\ref{sec:ts} takes $\OO(n^3)$ time.
Finally we check the three confluence conditions.
\begin{enumerate}
\item
For the first condition,
we check each of the $\OO(n^2)$ pairs of rules
$s_1 \circ s_2 \xr[\CCR]{} p, t_1 \circ t_2 \xr[\CCR]{} q$ with $p = q$,
$\TS(s_1 \circ s_2)$ and $\TS(t_1 \circ t_2)$
and check that $s_1 = t_1$ and $s_2 = t_2$.
\item
For the second condition, we consider the $\OO(n^2)$ triples
$s_1 \circ s_2 \to p \in \CCR$,
$t'$ with $p = (t')\EQU$,
and check that one of the $\OO(n)$ rules
$t_1 \circ t_2 \to q \in \EER$ with $q \to t' \in \RCR$
satisfies $s_1 = (t_1)\EQU$ and $s_2 = (t_2)\EQU$.
\item
Finally, for the third condition,
for all $\OO(n^2)$ pairs $s',t' \in \FF^{[]}$
with $(s')\EQU = (t')\EQU$
we check that $s' \downarrow t'$ holds.
\end{enumerate}
Overall these steps take $\OO(n^3)$ time.
Correctness follows from Lemma~\ref{lem:concon}.
\end{proof}

\begin{remark}
In our previous work~\cite{F12},
the congruence closure $\CCR$ was computed as a rewrite closure,
resulting in a congruence relation $\CCR'$
such that for all convertible terms $s$ and $t$,
$s \xr[\CCR' \cup \EER]{*} {\cdot} \xl[\CCR' \cup \EER]{*} t$.
In the present article, we instead work on the quotient obtained
by replacing each constant $[s]$ by $[s]/\CCR' = [s]\EQU$.
This results in some simplification,
because $\CCR$ is now deterministic,
but also in some complication,
because we are now working with two sets of extra constants,
$\FF^{[]}$ and $\FF\EQU$,
which have to be related using the map $(\cdot)\EQU$.
Apart from these differences,
the procedure presented here is identical to that presented in~\cite{F12}.
\end{remark}

%%%%%%%%%%%%%%%%%%%%%%%%%%%%%%%%%%%%%%%%%%%%%%%%%%%%%%%%%%%%%%%%%%%%%%%%%%%%%%
\section{Related Work}
\label{sec:rel}

% % % % % % % % % % % % % % % % % % % % % % % % % % % % % % % % % % % % % % %
\subsection{First-order Theory of Rewriting}
\label{sec:rel:dt}

The first-order theory of rewriting
was shown to be decidable by Dauchet and Tison~\cite{DT85}.
We can express the four properties of interest as follows.

\begin{align*}
{\xlr{*}} &\subseteq {\xr{*}} \cdot {\xl{*}}
\tag{\CR}
\\
{\xlr{*}} \cap (\TT \times \NF) &\subseteq {\xr{*}}
\tag{\NFP}
\\
{\xlr{*}} \cap (\NF \times \NF) &\subseteq {\equiv}
\tag{\UNC}
\\
({\xl{*}} \cdot {\xr{*}}) \cap (\NF \times \NF) &\subseteq {\equiv}
\tag{\UNR}
\end{align*}

Here, $\TT$ denotes the set of all ground terms,
$\NF$ represents the set of all ground normal forms
with respect to the given ground TRS,
and $\equiv$ stands for the identity relation on all terms.
We may assume that the maximum arity of the input TRS is at most two,
by currying the system first if necessary.
Thus, when following the procedure by Dauchet and Tison~\cite{DT85},
we can construct automata for the left-hand and right-hand sides
of the subset tests in polynomial time.
However, in order to perform the subset test,
the right-hand side needs to be represented by a deterministic automaton.
For \CR and \NFP this incurs an exponential cost,
but for \UNC and \UNR, the automaton is already deterministic.
Consequently, for the latter two properties,
a polynomial time decision procedure is obtained.

% % % % % % % % % % % % % % % % % % % % % % % % % % % % % % % % % % % % % % %
\subsection{Automation}

\begin{figure}
\begin{tikzpicture}
\begin{scope}[only marks, x=0.1cm, y=0.5cm]
\coordinate (O) at (0,10.5);
\pgfsetplotmarksize{0.15em}
\draw[->] (0,0) -- coordinate (x axis mid) (102,0);
\draw[->] (0,0) -- coordinate (y axis mid) (0,10.5);
\foreach \x in {0,20,40,60,80,100}
    \draw (\x,1pt) -- (\x,-3pt) node[anchor=north] {$\x\%$};
\foreach \y/\ytext in {0/0,2/2,4/4,6/6,8/8,10/{\geqslant}10}
    \draw (1pt,\y) -- (-3pt,\y) node[anchor=east] {$\ytext$};
\node[below=1cm,anchor=center] at (x axis mid) {solved};
\node[left=1cm,rotate=90,anchor=center] at (y axis mid) {time(s)};

%plots
\draw plot[mark=square*,   mark options={color=red,fill=white}] file {scatter/CSI-UNR.data};
\draw plot[mark=diamond*,  mark options={color=orange,fill=white}] file {scatter/CSI-UNC.data};
\draw plot[mark=+,        mark options={color=black,fill=white}] file {scatter/CSI-NFP.data};
\draw plot[mark=asterisk,        mark options={color=gray,fill=white}] file {scatter/CSI-CR.data};
\draw plot[mark=triangle*, mark options={color=blue,fill=white}] file {scatter/FORT-UNR.data};
\draw plot[mark=*,        mark options={color=green!70!black,fill=white}] file {scatter/FORT-UNC.data};
\draw plot[mark=x,        mark options={color=brown,fill=white}] file {scatter/FORT-NFP.data};
\draw plot[mark=pentagon,        mark options={color=cyan,fill=white}] file {scatter/FORT-CR.data};
\draw[dashed] (0,10) -- (100,10);
\draw[dashed] (100,0) -- (100,10);
\end{scope}

\footnotesize
\begin{scope}[shift={(0.5,4.5)}]
\draw (0,0) -- plot[mark=square*,   mark options={color=red,fill=white}] (0.25,0) -- (0.5,0) node[right]{\verb'CSI-UNR'};
\draw[yshift=-1\baselineskip] (0,0) -- plot[mark=diamond*,   mark options={color=orange,fill=white}] (0.25,0) -- (0.5,0) node[right]{\verb'CSI-UNC'};
\draw[yshift=-2\baselineskip] (0,0) -- plot[mark=+,   mark options={color=black,fill=white}] (0.25,0) -- (0.5,0) node[right]{\verb'CSI-NFP'};
\draw[yshift=-3\baselineskip] (0,0) -- plot[mark=asterisk,   mark options={color=gray,fill=white}] (0.25,0) -- (0.5,0) node[right]{\verb'CSI-CR'};
\draw[yshift=-4\baselineskip] (0,0) -- plot[mark=triangle*,   mark options={color=blue,fill=white}] (0.25,0) -- (0.5,0) node[right]{\verb'FORT-UNR'};
\draw[yshift=-5\baselineskip] (0,0) -- plot[mark=*,   mark options={color=green!70!black,fill=white}] (0.25,0) -- (0.5,0) node[right]{\verb'FORT-UNC'};
\draw[yshift=-6\baselineskip] (0,0) -- plot[mark=x,   mark options={color=brown,fill=white}] (0.25,0) -- (0.5,0) node[right]{\verb'FORT-NFP'};
\draw[yshift=-7\baselineskip] (0,0) -- plot[mark=pentagon,   mark options={color=cyan,fill=white}] (0.25,0) -- (0.5,0) node[right]{\verb'FORT-CR'};
\end{scope}

\end{tikzpicture}
\caption{Running times of \CSI and \FORT on 98 ground Cops}
\label{fig:experiments}
\end{figure}

Based on Dauchet and Tison's work,
the tool \FORT~\cite{RM16} implements decision procedures
for the first-order theory of rewriting
of left-linear, right-ground TRSs.
The algorithms described in this article are implemented in the
automated confluence prover \CSI~\cite{NFM17}.
In Figure~\ref{fig:experiments},
we compare the running times of \FORT to that of \CSI on the
98 ground TRSs in the Cops database%
\footnote{%
More information on Cops can be found at
\url{http://coco.nue.riec.tohoku.ac.jp/problems/}}
of \textbf{co}nfluence \textbf{p}roblem\textbf{s}.
(As the labels in Figure~\ref{fig:experiments} indicate,
we tested each of the for properties \UNR, \UNC, \NFP, and \CR,
for both tools, \CSI and \FORT.)
Most of the problems are easy for both tools,
but while \CSI never takes more than 0.5s on any of the input problems,
\FORT sometimes takes longer,
and even exceeds a timeout of 60 seconds on two of the confluence problems.
Full results are available online%
\footnote{%
\url{http://cl-informatik.uibk.ac.at/users/bf3/ground/statistic.php}
}.

% % % % % % % % % % % % % % % % % % % % % % % % % % % % % % % % % % % % % % %
\subsection{\UNR}

\UNR was shown to be decidable in polynomial time
by Verma~\cite{V09} and Godoy and Jaquemard~\cite{GJ09}.
As far as we know, these are the best previously published bounds.
The main focus of Verma's work is an abstract framework
for \UNC and \UNR that is also applicable to right-ground systems,
while Godoy and Jacquemard focus on \UNR for linear shallow TRSs.

% % % % % % % % % % % % % % % % % % % % % % % % % % % % % % % % % % % % % % %
\subsection{\UNC}

As far as we know,
the fastest previous algorithm for deciding \UNC of ground TRSs
is by Verma, Rusinovich and Lugiez~\cite{VRL01},
and takes $\OO(\Vert\RR\Vert^2\log\Vert\RR\Vert)$ time.
It is worth noting that the algorithm by Verma et al.\ 
is also based on the idea of using
a product construction for the intersection of the congruence closure
and an automaton recognizing normal forms.
However, the latter automaton is constructed in full,
leading to an essentially quadratic complexity bound.
Our algorithm is also closely related to another algorithm by Verma~\cite[Theorem~31]{V09},
but some care is needed to achieve an almost linear bound.

% % % % % % % % % % % % % % % % % % % % % % % % % % % % % % % % % % % % % % %
\subsection{Confluence}

\renewcommand{\FF}{\x F}
To derive a polynomial time decision procedure for confluence of ground TRSs,
Comon et al.~\cite{CGN01} use an approach based on a transformation by
Plaisted~\cite{P93} that flattens the TRS.
Then they test \emph{deep joinability} of sides of rules.
The authors sketch an implementation with complexity $\OO(n^5)$, where
$n$ is the size of the given TRS.
Tiwari \cite{Tiw02} and Godoy et al.~\cite{GTV04} base their approach on
a rewrite closure that constructs tree transducers---the given flattened TRS $\RR$
is converted into two TRSs $\FF$ and $\BB$ such that $\FF$ and $\BB^{-1}$
are left-flat, right-constant, $\FF$ is terminating,
and ${\xr[\RR]{*}} = {\xr[\FF]{*}} \cdot {\xr[\BB]{*}}$. They then
consider \emph{top-stabilizable} terms to derive conditions for confluence.
Tiwari obtains a bound of $\OO(n^9)$ (but a more careful implementation would
end up with $\OO(n^6)$), while Godoy et al.\ obtain a bound of $\OO(n^6)$.
The algorithm of \cite{CGN01} is limited to ground TRSs, but
\cite{Tiw02} extends the algorithm to certain shallow, linear systems,
and \cite{GTV03} treats shallow, linear systems in full generality.\footnote{%
The same claim can be found in \cite{GTV04}. However, rule splitting,
a key step in the proof of their Lemma~3.1, only works if left-hand side and
right-hand side variables are disjoint for every rule.}
In these extensions, however, the exponent depends on the maximum arity
of the function symbols of the given TRS. In our work we combine ideas
from \cite{CGN01,GTV04,Tiw02} in order to improve the complexity bound
to $\OO(n^3)$. The key ingredients are a Plaisted-style rewrite closure,
which results in TRSs $\FF$ and $\BB$ of only quadratic size
(namely $\EER \cup \RCR$ and $\EER^- \cup \RCR$), and
top-stabilizability, which is cheaper to test than deep joinability.

%%%%%%%%%%%%%%%%%%%%%%%%%%%%%%%%%%%%%%%%%%%%%%%%%%%%%%%%%%%%%%%%%%%%%%%%%%%%%%
\section{Conclusion}
\label{sec:end}

We have presented efficient polynomial time decision procedures
for deciding normal form properties and confluence of ground TRSs.
In particular, we showed how to decide
\UNR in $\OO(n^3 \log n)$ time,
\UNC in $\OO(n \log n)$ time, and
\NFP and \CR in $\OO(n^3)$ time,
where $n = \Vert \RR \Vert$ is the size of the given ground TRS.
As far as we know,
the bounds for the normal form properties are improvements on the literature;
most notably,
we obtain the first polynomial bound for deciding \NFP of ground TRSs.
The main innovation is the interleaving of
an enumeration of a potentially infinite set of witnesses
with checking conditions for the respective properties,
which also serve to ensure termination.
This is a theme that can be found in the procedures of
all three normal form properties presented here,
cf.~Listings~\ref{lst:unc}, \ref{lst:unr-1}, \ref{lst:unr-2} and~\ref{lst:nfp}.

There is a common theme in how the criteria are derived as well:
Starting with a conversion (or peak) that has a root step
(without loss of generality),
one uses the rewrite closure or congruence closure to restrict
which intermediate terms may be constants.
This opens the door to doing a top-down analysis of the intermediate terms,
ultimately giving rise to a finite criterion for the investigates property.
With the exception of \UNC, wich property translates directly to
the fact that a certain tree automaton accepts at most one term in any state,
the criteria arise from a careful case analysis
(taking inspiration from previous work on confluence of ground TRSs),
with little intuitive understanding.

As future work,
it may be interesting to investigate whether these ideas apply to
the larger class of left-linear,
right-ground systems as treated by Verma in~\cite{V09}.
It would also be interesting to prove correctness of the procedures
in an interactive theorem prover.

\subsection*{Acknowledgments}
The author is grateful to the anonymous reviewers for comments that
helped to improve the presentation.
Further thanks go to Franziska Rapp for \FORT,
which turned out to be an invaluable debugging tool,
to Aart Middeldorp for feedback on this article,
and to Vincent van Oostrom for moral support.

\bibliographystyle{plain}
\bibliography{long,confluence}

\providecommand{\noopsort}[1]{} \newcommand{\xdoi}[1]{doi:\,{\sffamily
  \href{http://dx.doi.org/#1}{\nolinkurl{#1}}}}
\begin{thebibliography}{10}

\bibitem{BN98}
F.~Baader and T.~Nipkow.
\newblock {\em Term Rewriting and All That}.
\newblock Cambridge University Press, 1998.
\newblock \xdoi{10.1017/CBO9781139172752}.

\bibitem{tata}
H.~Comon, M.~Dauchet, R.~Gilleron, F.~Jacquemard, D.~Lugiez, C.~L{\"o}ding,
  S.~Tison, and M.~Tommasi.
\newblock {\em Tree Automata Techniques and Applications}.
\newblock 2007.
\newblock URL: {\url{http://tata.gforge.inria.fr}}.

\bibitem{CGN01}
H.~Comon, G.~Godoy, and R.~Nieuwenhuis.
\newblock The confluence of ground term rewrite systems is decidable in
  polynomial time.
\newblock In {\em Proc.\ 42nd Annual Symposium on Foundations of Computer
  Science}, pages 298--307, 2001.
\newblock \xdoi{10.1109/SFCS.2001.959904}.

\bibitem{DT85}
M.~Dauchet and S.~Tison.
\newblock Tree automata and decidability in ground term rewriting systems.
\newblock In {\em Proc.\ 5th Conference on Fundamentals of Computation Theory},
  volume 199 of {\em Lecture Notes in Computer Science}, pages 80--84, 1985.
\newblock \xdoi{10.1007/BFb0028794}.

\bibitem{DG84}
W.F. Dowling and J.H. Gallier.
\newblock Linear-time algorithms for testing the satisfiability of
  propositional {H}orn formulae.
\newblock {\em Journal of Logic Programming}, 1(3):267--284, 1984.
\newblock \xdoi{10.1016/0743-1066(84)90014-1}.

\bibitem{F12}
B.~Felgenhauer.
\newblock Deciding confluence of ground term rewrite systems in cubic time.
\newblock In {\em Proc.\ 23rd International Conference on Rewriting Techniques
  and Applications}, volume~15 of {\em Leibniz International Proceedings in
  Informatics}, pages 165--175, 2012.
\newblock \xdoi{10.4230/LIPIcs.RTA.2012.165}.

\bibitem{GJ09}
G.~Godoy and F.~Jacquemard.
\newblock Unique normalization for shallow {TRS}.
\newblock In {\em Proc.\ 20th International Conference on Rewriting Techniques
  and Applications}, volume 5595 of {\em Lecture Notes in Computer Science},
  pages 63--77, 2009.
\newblock \xdoi{10.1007/978-3-642-02348-4_5}.

\bibitem{GTV03}
G.~Godoy, A.~Tiwari, and R.~Verma.
\newblock On the confluence of linear shallow term rewrite systems.
\newblock In {\em Proc.\ 20th International Symposium on Theoretical Aspects of
  Computer Science}, volume 2607 of {\em Lecture Notes in Computer Science},
  pages 85--96, 2003.
\newblock \xdoi{10.1007/3-540-36494-3_9}.

\bibitem{GTV04}
G.~Godoy, A.~Tiwari, and R.~Verma.
\newblock Deciding confluence of certain term rewriting systems in polynomial
  time.
\newblock {\em Annals of Pure and Applied Logic}, 130(1-3):33--59, 2004.
\newblock \xdoi{10.1016/j.apal.2004.04.005}.

\bibitem{K95}
S.~Kahrs.
\newblock Confluence of curried term-rewriting systems.
\newblock {\em Journal of Symbolic Computation}, 19(6):601--623, 1995.
\newblock \xdoi{10.1006/jsco.1995.1035}.

\bibitem{KKSV96}
R.~Kennaway, J.W. Klop, {M.~Ronan} Sleep, and {F.-J.~de} Vries.
\newblock Comparing curried and uncurried rewriting.
\newblock {\em Journal of Symbolic Computation}, 21(1):15--39, 1996.
\newblock \xdoi{10.1006/jsco.1996.0002}.

\bibitem{M96}
M.~Marchiori.
\newblock On the modularity of normal forms in rewriting.
\newblock {\em Journal of Symbolic Computation}, 22(2):143--154, 1996.
\newblock \xdoi{10.1016/jsco.1996.0045}.

\bibitem{M90}
A.~Middeldorp.
\newblock {\em Modular Properties of Term Rewriting Systems}.
\newblock PhD thesis, Vrije Universiteit Amsterdam, 1990.

\bibitem{NFM17}
J.~Nagele, B.~Felgenhauer, and A.~Middeldorp.
\newblock {CSI}: New evidence -- a progress report.
\newblock In {\em Proc.\ 26th International Conference on Automated Deduction},
  volume 10395 of {\em Lecture Notes in Artificial Intelligence}, pages
  385--397, 2017.
\newblock \xdoi{10.1007/978-3-319-63046-5_24}.

\bibitem{NO80}
G.~Nelson and D.C. Oppen.
\newblock Fast decision procedures based on congruence closure.
\newblock {\em Journal of the ACM}, 27(2):356--364, 1980.
\newblock \xdoi{10.1145/322186.322198}.

\bibitem{NO07}
R.~Nieuwenhuis and A.~Oliveras.
\newblock Fast congruence closure and extensions.
\newblock {\em Information and Computation}, 205(4):557--580, 2007.
\newblock \xdoi{10.1016/j.ic.2006.08.009}.

\bibitem{P93}
D.~Plaisted.
\newblock Polynomial time termination and constraint satisfaction tests.
\newblock In {\em Proc.\ 5th International Conference on Rewriting Techniques
  and Applications}, volume 690 of {\em Lecture Notes in Computer Science},
  pages 405--420, 1993.
\newblock \xdoi{10.1007/3-540-56868-9_30}.

\bibitem{RM16}
F.~Rapp and A.~Middeldorp.
\newblock Automating the first-order theory of left-linear right-ground term
  rewrite systems.
\newblock In {\em Proc.\ 1st International Conference on Formal Structures for
  Computation and Deduction}, volume~52 of {\em Leibniz International
  Proceedings in Informatics}, pages 36:1--36:12, 2016.
\newblock \xdoi{10.4230/LIPIcs.FSCD.2016.36}.

\bibitem{Tiw02}
A.~Tiwari.
\newblock Deciding confluence of certain term rewriting systems in polynomial
  time.
\newblock In {\em Proc.\ 17th IEEE Symposium on Logic in Computer Science},
  pages 447--457, 2002.
\newblock \xdoi{10.1109/LICS.2002.1029852}.

\bibitem{T87}
Y.~Toyama.
\newblock On the {C}hurch-{R}osser property for the direct sum of term
  rewriting systems.
\newblock {\em Journal of the ACM}, 34(1):128--143, 1987.
\newblock \xdoi{10.1145/7531.7534}.

\bibitem{V09}
R.~Verma.
\newblock Complexity of normal form properties and reductions for term
  rewriting problems.
\newblock {\em Fundamenta Informaticae}, 92(1--2):145--168, 2009.
\newblock \xdoi{10.3233/FI-2009-0070}.

\bibitem{VRL01}
R.M. Verma, M.~Rusinowitch, and D.~Lugiez.
\newblock Algorithms and reductions for rewriting problems.
\newblock {\em Fundamenta Informaticae}, 46(3):257--276, 2001.

\end{thebibliography}

\end{document}